\documentclass[11pt]{article}
\usepackage{graphicx} 
\usepackage[english]{babel}
\usepackage{csquotes}
\usepackage[top=2cm,bottom=3.5cm,left=2.5cm,right=2.5cm]{geometry}
\usepackage{amsmath, amsfonts, amsthm, algorithm, algorithmic, enumerate, mathtools, bm, enumitem, braket}
\usepackage{tikz-cd}
\usepackage{tikz}
\usepackage{hyperref}
\hypersetup{colorlinks,allcolors=blue,breaklinks=true}
\usepackage{breakurl}
\usepackage[backend=bibtex, style=alphabetic, maxbibnames=99, maxalphanames=10]{biblatex}
\usepackage{multirow}

\newtheorem{definition}{Definition}[section]
\newtheorem{notation}[definition]{Notation} 
\newtheorem{construction}[definition]{Construction}
\newtheorem{remark}[definition]{Remark}
\newtheorem{obs}[definition]{Observation}
\newtheorem{theorem}{Theorem}[section]
\newtheorem{corollary}[theorem]{Corollary}
\newtheorem{lemma}[theorem]{Lemma}

\newtheorem{proposition}[theorem]{Proposition}
\newtheorem{example}{Example}[section]

\newcommand{\coeff}{\mathtt{coeff}}
\newcommand{\YES}{\mathtt{YES}}
\newcommand{\NO}{\mathtt{NO}}
\DeclareMathOperator*{\E}{\mathbb{E}}
\newcommand{\N}{\mathbb{N}}

\newcommand{\Z}{\mathbb{Z}}
\newcommand{\R}{\mathbb{R}}

\newcommand{\I}{\mathcal{I}}

\newcommand{\G}{\mathcal{G}}
\newcommand{\K}{\mathcal{K}}

\newcommand{\poly}{\texttt{poly}}
\newcommand{\innerproduct}[2]{\langle #1, #2 \rangle}

\newcommand{\abs}[1]{\left\lvert #1 \right\rvert}

\newcommand{\norm}[1]{\left\lVert #1 \right\rVert}

\newcommand{\LIN}[2]{{#1}\mbox{-}\text{LIN}(#2)}
\newcommand{\ALIN}[2]{{#1}\mbox{-}\text{ALIN}(#2)}
\newcommand{\outdeg}{\text{out-deg}}

\newcommand{\val}{\texttt{val}}
\newcommand{\algval}{\texttt{alg-val}}

\DeclarePairedDelimiter\floor{\lfloor}{\rfloor}
\DeclarePairedDelimiter\ceil{\lceil}{\rceil}

\begin{document}
\title{Spectral Method attacks Sparse LWE, Sparse LPN and Beyond}
\date{\vspace{-4ex}}
\author{
Shashwat Agrawal\\ Indian Institute of Technology Delhi  \\ \texttt{csz248012@iitd.ac.in} \and
Amitabha Bagchi\\ Indian Institute of Technology Delhi \\ \texttt{bagchi@cse.iitd.ac.in} \and
Rajendra Kumar\\ Indian Institute of Technology Delhi  \\ \texttt{rajendra@cse.iitd.ac.in}}
\maketitle             
\begin{abstract}
Given a set of $k$-sparse linear equations over a ring $R$, we give algorithms to determine whether the right-hand sides are random or have a secret assignment planted with noise. For a parameter $k/2\leq l\leq n$, we give a spectral method to solve this problem in $\widetilde{O}\left(\binom{n}{l}\abs{R}^l\right)$ time except with probability at most $n^{-\Omega(l)}$, provided the number of samples is roughly at least $\left(\frac{\abs{R}n}{l}\right)^{k/2}$. This attack generalizes the Kikuchi method described by Wein et. al. (Journal of the ACM 2019) for $\Z_2$ to (commutative) rings of any finite size. We also give a simpler algorithm with better runtime than the spectral method and better sample complexity when $\abs{R}=\omega(n/l)$. As a consequence, we obtain new sample-time tradeoffs for the decision problem of sparse LWE, sparse LPN over higher modulus $q$, and in general the distinguishing random vs planted $\Z_q$-linear equations for a large class of noise distributions. Our results imply a tightness of the hardness claims of Jain, Lin, Saha (Annual International Cryptology Conference, 2024) for sparse LWE.
\end{abstract}

\clearpage

\section{Introduction} \label{intro}

Solving systems of linear equations is one of the most fundamental problems in mathematics and computer science. Over a field, we can efficiently solve this via Gaussian elimination in polynomial time. The problem, however, becomes considerably harder when the equations are \emph{noisy}: given $m$ equations of the form $\langle\bm{a}_i, \bm{s}\rangle +e_i=b_i$  where $e_i$ is an unknown error term, recovering the secret $\bm{s}$ is no longer straightforward. The difficulty of solving these noisy linear equations  depends on the structure of the error distribution, and understanding precisely \emph{when} efficient recovery is possible remains a central question in both learning theory and cryptography.

Two of the most well-studied average-case instantiations of noisy linear equations are the Learning With Errors (LWE) problem, introduced by Regev~\cite{reg05}, and its close cousin, the Learning Parity with Noise (LPN) problem~\cite{alekhnovich2003more}; both serve as central hardness assumptions in modern cryptography. In LWE, one samples $(\bm{a}, b) \in \mathbb{Z}_q^n \times \mathbb{Z}_q$ by choosing $\bm{a} \leftarrow \mathbb{Z}_q^n$ uniformly, an error $e \leftarrow \Psi$ (typically a discrete Gaussian of width $r < q$), and setting $b = \langle \bm{a}, \bm{s} \rangle + e \bmod q$. The decision problem is to distinguish such samples from uniform over $\mathbb{Z}_q^n \times \mathbb{Z}_q$. The LPN problem is most commonly studied over $\mathbb{Z}_2$ with Bernoulli noise, or more generally over modulus $q$ with sparse (low Hamming-weight) noise, and has a closely related structure.
Most lattice-based cryptosystems rely on LWE's conjectured intractability against classical and quantum algorithms~\cite{reg05,peik09,gentry2009fully,brakerski2014lattice,crypto_primitives_on_LWE}, while LPN's hardness rests on decades of cryptanalysis~\cite{goldreich2000candidate,cryan2001pseudorandom,feige2002relations,mossel2003epsilon,allen2015refute,applebaum2016algebraic,kothari2017sum}, offering competitive efficiency despite fewer constructions~\cite{lpn_crypto,applebaum2017secure}.

Variants of these problems are considered to balance security with efficiency of cryptographic primitives. Different choices of modulus, noise distribution, or \textsf{structural constraints} (such as sparsity) can lead to primitives with faster key generation, smaller ciphertexts, or reduced computational overhead, while still relying on problems which are believed to have same hardness~\cite{lyubashevsky2013ideal,PKE_from_sparse_LPN}.
In this work, we focus on LWE and LPN with sparsity constraints, where the coefficient vectors $\bm{a}\in\mathbb{Z}_q^n$ are required to have only $k$ non-zero entries for a parameter $k<n$. The motivation for this sparsity is that the samples can be computed and stored with a roughly $O(n/k)$ factor improvement in efficiency. Hardness of sparse LPN over $\Z_2$ and variants have been studied for as long as forty years in average-case complexity and have intimate connections to constraint satisfaction problems \cite{oldest_ref_for_sparse_LPN,sparse_LPN_ref2,sparse_LPN_ref3,sparse_LPN_ref4,sos_lower_bounds}. Sparse variants of LWE/LPN over large moduli 
$q$ (where $q$ is at least polynomial in the security parameter) are useful for designing efficient advanced cryptographic primitives, including FHE but remain underexplored. 
We summarize the current known attacks for sparse variants below.

The task in analyzing the hardness of these sparse variants is to determine the computational complexity as a function of $m$, i.e. the number of samples $(\bm{a},b)$. For $m\leq n$, the random and the planted noisy instances are statistically indistinguishable. Raghavendra, Rao and Schramm \cite{sos_upper_bound} showed that sparse LPN over $\Z_2$ can be solved in $n^{O(l)}$ time when $m$ is a quantity of the form $l\left(\frac{n}{l}\right)^{k/2}\log^{c(k)}(n)$ using a degree $l$ ``Sum of squares'' method for $l\geq k/2$. An improved (and the best known) bound of the form $l\left(\frac{n}{l}\right)^{k/2}\log{n}$ for $m$ with the same time complexity can be obtained using the ``Kikuchi method'' by Wein, Alaoui, and Moore \cite{kikuchi_ref1} (independently by Hastings \cite{kikuchi_ref2}).\footnote{The Kikuchi method is inspired from Kikuchi free energy, a marginal approximation to Gibbs free energy \cite{kikuchi1951theory}} Schmidhuber, O’Donnell, Kothari, and Babbush \cite{quartic_speedup} gave a quantum algorithm that achieves a quartic speedup in time (i.e. $n^{l/4}\cdotp\text{poly}(n)$) over the Kikuchi method for certain choices of the sample complexity $m$, while using exponentially less space.

Despite the presence of numerous other variants of LWE, sparse LWE was not studied until a recent work by Jain, Lin, and Saha \cite{sparse_LWE}. Assuming hardness of sparse LWE, they construct secret key Linearly Homomorphic Encryption (LHE) schemes with slightly super-constant/constant overheads, and secret key homomorphic encryption for arbitrary constant-degree polynomials with slightly super-constant/constant overheads. They also perform some cryptanalysis on sparse LWE by suggesting two attacks for sparse LWE, namely the sparse vector in kernel attack and the dense minor attack. On analyzing the feasibility of these attacks for $k\geq\Omega(\log n)$ and $m = \text{poly}(n)$, it turns out that these attacks will take exponential time in $n$ with high probability. Hence, the authors conjecture that sparse LWE with sparsity $k\geq\Omega(\log n)$ and polynomially many samples $m = \text{poly}(n)$ is as hard as the standard LWE with dimension $n'=\Theta(n)$ and polynomially many samples. 

Subsequently, Bangachev et al.~\cite{bangachev2025near} give improved reductions from dense LWE to sparse LWE, and similarly from dense LPN to sparse LPN. Their reductions aim to establish $n^{\Omega(k)}$-hardness of sparse LWE under the assumption that dense LWE requires exponential time, but still incur an additional $1 / \log k$ multiplicative factor in the exponent.

\textbf{Our Contributions.} In this work we provide algorithms for solving the analogue of sparse LPN over general rings $R$ with $\abs{R}>2$, which we call $\LIN{k}{R}$. Here the task is to determine whether a system of $k$-sparse linear equations over a finite ring $R$ (having randomly chosen left-hand sides) has random right-hand sides or having a planted secret assignment with a noise distribution $\Psi$. We obtain analogous results of tradeoff between the sample complexity and runtime as those for sparse LPN. 

\begin{theorem}\label{thm:Kikuchi attack informal}
    (Informal version of Corollary~\ref{solving kLIN small q}) There exists an algorithm to distinguish whether a $\LIN{k}{R}$ instance on $n$ dimensions is random or planted in $\widetilde{O}\left(\binom{n}{l}(\abs{R}-1)^l\right)$ time for some parameter $\ceil{k/2}\leq l\leq n-\floor{k/2}$, provided the number of samples is $m \geq (\abs{R}-1)^{\floor{k/2}}\left(\frac{n}{l}\right)^{\ceil{k/2}}l\ln((\abs{R}-1)n)$.
\end{theorem}
\begin{remark}
We note that the attack In the above theorem can achieve a speedup in runtime using quantum computation via the same algorithm as given by \cite{quartic_speedup} for $R = \Z_2$. This speedup can be between quadratic to quartic depending on the choice of parameter $l$.
\end{remark}

We also have a simpler algorithm for the distinguishing problem with a better time complexity, and also a better sample complexity than above when $\abs{R}=\omega(n/l)$. 

\begin{theorem}\label{thm:simpler algo informal}
    (Informal version of Theorem~\ref{thm:simple solving kLIN}) There exists an algorithm to distinguish whether a $\LIN{k}{R}$ instance on $n$ dimensions is random or planted in $\widetilde{O}\left((n/l)^{k+1}+\abs{R}^l\right)$ time for some parameter $k\leq l\leq n$, provided the number of samples is $m \geq \left(\frac{n}{l}\right)^{k}l\ln((\abs{R}-1)n)$.
\end{theorem}

As can be seen, the parameter $l$ in the above theorem trades sample complexity for time. In particular, we obtain tradeoffs between the sample complexity $m$ and time complexity for sparse LWE and sparse LPN over higher modulus $q$. As an interesting case, if either $k=O(\log n/\log\log \abs{R})$ or we are allowed $m = n^{\Omega(\log\log \abs{R})}$ samples, we can choose $l$ appropriately so that the time complexity is $\exp({\delta n})$ for arbitrarily small $\delta > 0$. Assuming standard LWE does not have an arbitrarily small exponential time, our result also establishes the tightness of the conjecture on the hardness of sparse LWE by \cite{sparse_LWE}.

Moreover, the algorithms in the above theorems are applicable for a wide class of noise distributions $\Psi$. For the ring $\Z_q$, apart from the discrete Gaussian noise in LWE or the `low Hamming weight' noise in LPN, the distinguishing problem is solvable for a large class of other plausible noise distributions. For example, one might consider a distribution where zero error is allowed only with probability $1/q$, otherwise the error is rejection sampled again according to the discrete Gaussian. Another example can be an addition of noises, one from the discrete Gaussian and another from the LPN distribution noise, which was recently proposed in a bid for new hardness on noisy linear equations \cite{ghosal2025post}. In the sparse regime, our algorithms attack the problem for these noise distributions as well.

For Theorem~\ref{thm:Kikuchi attack informal} we extend the Kikuchi method used by \cite{kikuchi_ref1} for $\LIN{k}{\Z_2}$. This involves formulating a large graph corresponding to the $\LIN{k}{R}$ instance, called the Kikuchi Graph. We determine whether the instance is random or planted based on the spectral norm of the adjacency matrix of this graph. For Theorem~\ref{thm:simpler algo informal}, we collect all subsamples of the original instance that have non-zero coefficients only on the first $l$ variables, and then distinguish by naively checking all $\abs{R}^l$ assignments on this sub-instance.

In \cite{sparse_LWE}, Jain, Lin and Saha had defined sparse LWE in a manner which allows equations to be \emph{at most} $k$-sparse instead of exactly $k$-sparse \footnote{although in their Introduction, they mention that an equation has exactly $k$ non-zero entries}. We analyse this variant for general rings as well, which we call $\ALIN{k}{R}$. Intuitively, the distinguishing problem should be easier in this case, as left-hand sides here can have very few or even no non-zero entries. For large $\abs{R}$, the above approach itself can work, however when $\abs{R}$ is small, we can achieve a much better complexity.

\begin{theorem}
    (Informal version of Theorem~\ref{thm:simple solving ALIN small q}) There exists an algorithm to distinguish whether a $\ALIN{k}{R}$ instance on $n$ dimensions is random or planted in $\widetilde{O}\left(\abs{R}^k\right)$ time, provided the number of samples is $m \geq \widetilde{\Omega}\left(\abs{R}^k\right)$.
\end{theorem}

\cite{sparse_LWE} had also given a reduction from standard LWE on $k$ dimensions to (at most) $k$-sparse LWE on $n\geq k$ dimensions. It is believed that distinguishing standard LWE takes $2^{\Omega(k)}$ time given arbitrary many samples. Our result above can be seen as a step towards matching this hardness result for (at most) $k$-sparse LWE.

\textbf{Independent Work.} Independent of our work, Kocurek and Manohar \cite{kocurek2025spectral} have studied the problem of refuting random and semirandom instances of $\LIN{k}{R}$\footnote{a semirandom instance is when the left-hand sides of a $\LIN{k}{R}$ instance can be arbitrarily chosen, but the right-hand sides are uniformly and independently chosen}. For this task, they define Kikuchi graphs for finite fields and abelian groups that are similar to our construction. Their sample-time tradeoffs are also similar to ours. However, our objective is different from the refutation task; ours is the task of distinguishing a random vs a planted instance. It turns out that the statistic they use for the refutation task, the  {\em value} of an instance (maximum fraction of satisfiable constraints), cannot be extended to the distinguishing task when the noise distribution for the planted case has a high probability of introducing an error, i.e., has a low probability of making zero error. The statistic we use, which we call the \emph{advantage} (Definition~\ref{def advantage}), is able to handle error distributions which satisfy a much milder condition than the one~\cite{kocurek2025spectral} need, and is thereby more widely applicable.

Moreover, they construct separate Kikuchi constructions depending on the sparsity $k$ being even or odd, and also depending on the ring being a field or non-field, with significantly more involved construction in the odd sparsity case. 

Our Kikuchi construction is able to mitigate both these issues; it is more amenable to the distinguishing task, and is a unified and simplified construction for all rings $R$ and both sparsity parities when $\abs{R}>2$. Further, we are able to construct a simpler distinguisher which has a significantly lower runtime than that of their simple refutation algorithm when the size of the ring $R$ is large. 

We elaborate on these aspects in Subsection~\ref{subsec:compare with KM25}.

\textbf{Organization.} Section~\ref{sec:overview} gives a high level overview of our spectral Kikuchi method. Section~\ref{sec:preliminaries} introduces some basic notations, standard results, definitions of the $\LIN{k}{R}$ problem and variants. Section~\ref{sec:Kikuchi graph} describes our Kikuchi graph construction for $\LIN{k}{R}$. In Section~\ref{sec:Kikuchi spectral}, we utilize the spectral norm of the Kikuchi graph to distinguish planted vs random $\LIN{k}{R}$ instances. In Section~\ref{sec:simple algo}, we describe our simpler algorithm for distinguishing $\LIN{k}{R}$ instances. In Section~\ref{sec:alin} we deal with the `at most' $k$-sparse version, called $\ALIN{k}{R}$, and give a simple algorithm to distinguish a planted vs random version of this problem. Finally, in Section~\ref{sec:discussion} we discuss how our techniques and results compare with the works \cite{kocurek2025spectral} and \cite{sparse_LWE}.

\section{Overview of the Kikuchi Method}\label{sec:overview}
 
In this section, we provide a high-level overview of our Kikuchi construction (Section~\ref{sec:Kikuchi graph}) and the spectral method to distinguish planted vs random $\LIN{k}{R}$ instances (Section~\ref{sec:Kikuchi spectral}). We first summarize how this technique was used for the $k$-XOR problem ($k$-sparse LWE/LPN for $\Z_2$), \cite{kikuchi_ref1} (and independently \cite{kikuchi_ref2}), and then discuss the modifications we needed to make it work for general rings. 

Let us review the steps of the Kikuchi method for $\Z_2$:

\begin{enumerate}
\item Each $k$-sparse equation is converted into a set of $2$-sparse equations, i.e., equations in two variables. This is done by assigning variables to $l$-sized subsets of $[n]$ where $l\geq k/2$ is a parameter. For each original equation of the form $C: \sum_{j\in S}{x_j}=b$, we write down equations $C': y_T + y_{T'}=b$ for all possible $l$-sized subsets $T,T'\subseteq [n]$ such that their symmetric difference equals $S$. For every vector $\bm{s}\in\Z_2^n$, one can define a natural vector $\bm{z}\in\Z_2^{\binom{[n]}{l}}$ such that $\bm{s}$ satisfies $C$ iff $\bm{z}$ satisfies $C'$.
\item We note that such sets of $2$-sparse equations correspond to a graph with $\binom{[n]}{l}$ as nodes and an edge between each pair of variables that appear in an equation. The weight of the edge is the value $(-1)^b$, where $b$ is the RHS of the equation. 
\item For the $k$-XOR problem an assignment $\bm{x}\in\Z_2^n$ has a function called the \emph{advantage} (or \emph{value}) associated with it. This advantage is defined in such a way that it is equal (up to scaling) to the quadratic form associated with the adjacency matrix of the graph defined in the previous item. Hence the optimal advantage can be directly related to the spectral norm of the adjacency matrix.
\item Distinguishing an instance as random or planted is done on the basis of spectral norm of the adjacency matrix of this graph. Appropriate tail bounds on the spectral norm are proved to show that it will be unambiguously separated in the two cases. The spectral norm of the expected adjacency matrix is $0$ in the random case, and in the planted case is of the order of the average degree of the graph, which in turn is proportional to the number of samples of the $k$-XOR instance. Thus in order to get tail bounds which are separated in the two cases, it is essential that there are sufficiently many samples. 
\end{enumerate}

Now we discuss the specific modifications we made to generalize this for a ring $R$:

\begin{enumerate}[label=\roman*.]
    % \item Linear equations over $R$ are rewritten as product equations over $\Omega_q = \{1,\omega_q,\ldots,\omega_q^{q-1}\}$, where $\omega_q = e^{2\pi\iota/q}$ is a primitive $q\textsuperscript{th}$ root of unity.  
    \item To successfully convert a $k$-sparse equation to a collection of specific $2$-sparse equations, there are a couple of challenges in the $\abs{R}>2$ regime that we mitigate. Firstly, the LHS of an equation here consists of a $k$-sparse vector $\bm{a}\in R^n$ instead of just the non-zero indices for $\Z_2$. Accordingly, the 2-sparse equations will have to be defined over variables indexed by the subset of $R^n$ which has $l$-sparse vectors. Secondly, for defining the $2$-sparse constraints in this space, we will generalize the notion of choosing pairs satisfying a symmetric difference: we choose all pairs $(\bm{u},\bm{v})$ such that $\bm{v}-\bm{u} = \bm{a}$.
    \item We want to have $\mathbb{C}$-valued signs for each edge pair, so that we could utilize spectral method for random matrices. So the sign of $(\bm{u},\bm{v})$ is $\chi(b)$ for the canonical character map of (the underlying additive group of) $R$, $\chi : R \to \mathbb{C}$. Our generalization of symmetric difference is asymmetric and yields a directed graph. To apply the spectral method, it is desirable that the adjacency matrix of the graph be \emph{Hermitian}. We enable this by adding edges $\bm{v}\xrightarrow{\chi(-b)}\bm{u}$ for every original edge $\bm{u}\xrightarrow{\chi(b)}\bm{v}$ in the graph. 
    \item We define a new notion of advantage of an assignment. For $R=\Z_q$ this captures the `extent' to which the assignment violates a constraint (unlike the $\Z_2$ case where this can only capture whether a constraint got satisfied). We then show that it is equal (up to scaling) to the quadratic form of the (Hermitian) matrix from the previous item. 
    \item Like the $\Z_2$ case, we can compute the spectral norm of the adjacency matrix, and prove appropriate tail bounds to ensure a separation between the spectral norms in the random vs planted instances. The method we employed for the tail bounds slightly deviates for that in the case of $\Z_2$, say in \cite{quartic_speedup}. In the latter, the upper tail bound for the random case was obtained using Matrix Chernoff bounds for Rademacher series, and the lower tail bound for the planted case was obtained using properties of Skellam distribution. We use the \emph{Matrix Bernstein inequality} to get both the tail bounds. 
\end{enumerate} 

\section{Preliminaries}\label{sec:preliminaries}

\subsection{Notations}
Let $\mathbb{N}, \R, \mathbb{C}$ denote the set of naturals, reals, complex numbers respectively. Let $\iota = \sqrt{-1}$. For $n\in\N$, $[n] := \{1,2,\ldots,n\}$. For $0\leq r\leq n$, $\binom{[n]}{r}$ denotes the set of all $r$-size subsets of $[n]$. For a set $S$, $\text{Unif}(S)$ denotes the uniform distribution on $S$. For $a\in\mathbb{C}$, $\text{Re}(a)$ denotes its real part, $a^*$ denotes its complex conjugate. $\omega_q := e^{\frac{2\pi\iota}{q}}$ denotes the $q$\textsuperscript{th} primitive root of unity. 

For a vector $\bm{x}$, unless stated otherwise, $x_j$ denotes the $j$\textsuperscript{th} coordinate of $\bm{x}$. For a matrix $A$, $\norm{A}$ denotes its spectral norm. $A^{\dagger}$ denotes the conjugate transpose of $A$. $\bm{O}$ denotes a null matrix, and $\bm{0}$ denotes a null vector.

For a function $f:\N\to\N$, the notation $\widetilde{O}(f(n))$ denotes $O(f(n)\cdot\poly(\log f(n)))$ for some polynomial $\poly(\cdot)$.

In this paper, a ring $R$ is assumed to be finite, commutative, and to have a multiplicative identity $1\in R$. For $t\leq n$, $R^n_t$ denotes the subset of $R^n$ consisting of all vectors with exactly $t$ coordinates in $R\setminus\{0\}$, and $R^n_{\leq t}$ is the subset of $R^n$ consisting of all vectors with at most $t$ coordinates in $R\setminus\{0\}$. For $\bm{x},\bm{y}\in R^n$, $\innerproduct{\bm{x}}{\bm{y}} := x_1y_1 + \ldots +x_ny_n$ denotes the sum of their coordinate-wise products (the ring product operation). 

For a vector $\bm{x} \in R^n$, we denote $\bm{x}^{\odot l}\in R^{R^n_l}$ as the vector whose $\bm{u}\textsuperscript{th}$ coordinate (for $\bm{u}\in R^n_l$) is given by $x^{\odot l}_{\bm{u}} := \innerproduct{\bm{u}}{\bm{x}}$.

For a directed graph $\G$, let $\text{out-deg}_{\G}(T)$ denote the number of outgoing edges of a vertex $T$ of $\G$, and let $D(\G)$ denote the maximum outgoing degree of $\G$. We may skip writing $\G$ in these notations when the graph is clear in the context.

\subsection{Canonical Character of an Abelian Group}
For a ring $R$, its underlying additive group is finite and abelian, so it can be assumed to be group isomorphic to $\Z_{q_1}\times\ldots\times\Z_{q_t}$. Let $\psi : (R,+) \to \Z_{q_1}\times\ldots\times\Z_{q_t}$ be this group isomorphism, which we can find in $\widetilde{O}(\abs{R})$ time. Define the canonical character of $R$ as the group homomorphism
\begin{align*}
    \chi_R : &(R,+) \to \left(\mathbb{C},\cdot\right) \\
    &x \mapsto \omega_{q_1}^{\psi(x)_1}\cdots\omega_{q_t}^{\psi(x)_t}
\end{align*}

We write $\chi_R$ as simply $\chi$ when the ring is clear in the context.

\begin{obs}\label{obs:sum of chi}
$\sum_{a\in R}\chi(a) = 0$.
\end{obs}

\subsection{Concentration Inequalities}

Let us state tail bound results for random variables (both scalar and matrix) which we will use in this work.

\begin{theorem}\label{thm:chernoff}
    (Chernoff Bound) Let $X = \sum_{i=1}^n X_i$ be the sum of independent random variables on $[0,1]$ with $E[X] = \mu$. Then

\begin{align*}
    \forall\ \delta\geq 0,\ &\Pr\left[X \ge (1+\delta)\mu\right] \le e^{-\frac{\delta^2\mu}{2+\delta}} \\
    \forall\ 0\leq \delta\leq 1,\ &\Pr\left[X \le (1-\delta)\mu\right] \le e^{-\frac{\delta^2\mu}{2}} 
\end{align*}

\end{theorem}

\begin{theorem}\label{thm:Hoeffding}
    (Hoeffding Inequality) Let $X_1, \dots, X_n$ be independent and identically distributed random variables on $[a, b]$, let $\overline{X}= \frac{1}{n}\sum_{i\in[n]}X_i$, and let $\mu := \E[\overline{X}]$. Then
\begin{align*}
    &\Pr\left[\overline{X}\leq \mu-t\right] < \exp\left( -\frac{2nt^2}{(b-a)^2} \right) \\
    &\Pr\left[\overline{X}\geq \mu+t\right] < \exp\left( -\frac{2nt^2}{(b-a)^2} \right)
\end{align*}
\end{theorem}

The Matrix Bernstein Inequality bounds the norm of a sum of random Hermitian matrices (see, e.g., \cite{tropp_matrix_bounds}). 

\begin{theorem} \label{matrix bernstein} (Theorem 1.4 in \cite{tropp_matrix_bounds})
Consider a finite sequence $\{\bm{X}_i\}$ of independent, random, Hermitian matrices with dimension $N$. Assume that each random matrix satisfies $\mathbb{E}\left[\bm{X}_i\right] = \bm{O}$ and $\norm{\bm{X}_i} \leq B$ with probability 1.  
Then, for all $t \geq 0$,
\[
\mathbb{P} \left[ \norm{ \sum_{i} \bm{X}_i } \geq t \right] 
\leq N \cdot \exp\!\left( \frac{-t^{2}/2}{\sigma^{2} + (B t)/3} \right),
\]
where 
\[
\sigma^{2} := \left\| \sum_{i} \mathbb{E}\left[\bm{X}_i^{2}\right] \right\|.
\]
\end{theorem}

\subsection{$\LIN{k}{R}$}

We now state definitions and problems related to sparse systems of linear equations over a ring.

\begin{definition}\label{defn: kLIN}
    (exact $k$-sparse instance) Let $R$ be a ring. A $\LIN{k}{R}$ instance $\I = \{(\bm{a}_i,b_i)_{i\in[m]}\}$ consists of vectors $\bm{a}_i\in R^n_{k}$ and $b_i\in R$.
\end{definition}

\begin{definition}\label{defn: at most kLIN}
    (at most $k$-sparse instance) Let $R$ be a ring. A $\ALIN{k}{R}$ instance $\I = \{(\bm{a}_i,b_i)_{i\in[m]}\}$ consists of vectors $\bm{a}_i\in R^n_{\leq k}$ and $b_i\in R$.
\end{definition}

For a $\LIN{k}{R}$ or $\ALIN{k}{R}$ instance $\I$, we view $(\bm{a},b)\in\I$ as a $R$-linear homogeneous equation in $n$ variables $x_1,\ldots,x_n$: $\sum_{j\in[n]}{a_{j}x_j} = b$.

\begin{remark}
    The exact $k$-sparse instance has been defined and analysed in \cite{kocurek2025spectral} for abelian groups and in numerous other works for $\Z_2$ (e.g. \cite{oldest_ref_for_sparse_LPN,sparse_LPN_ref2,sparse_LPN_ref3,sos_upper_bound,kikuchi_ref1}). On the other hand, \cite{sparse_LWE} defined sparse LWE as an \emph{at most} $k$-sparse instance over $\Z_p$ for prime $p$. We consider and analyse both versions, and compare with their works in Section~\ref{sec:discussion}.
\end{remark}

Next, we define a notion of \emph{advantage} of an assignment for an instance.

\begin{definition}\label{def advantage}
    Given a $\LIN{k}{R}$ (or $\ALIN{k}{R}$) instance $\I$ and an assignment $\bm{x}\in R^n$, the advantage of assignment $\bm{x}$ for $\I$ is given by
    $$\texttt{adv}_{\I}(\bm{x}) := \frac{1}{\abs{\I}}\sum_{(\bm{a},b)\in\I}\text{Re}\left(\chi\left(b-\innerproduct{\bm{a}}{\bm{x}}\right)\right) \in [-1,1]$$
    We denote $\texttt{adv}(\I) := \max\limits_{\bm{x}\in R^n}{\abs{\texttt{adv}_{\I}(\bm{x})}}$.
    % It can be seen that 
    % $$\text{adv}_{\I}(\bm{x}) := \frac{m_{sat}(\bm{x}) - (1/2)m_{unsat}(\bm{x})}{m}$$
    % where $m_{sat}(x)$ denotes the number of constraints satisfied by $x$ and $m_{unsat}(x)$ denotes the number of constraints violated by $x$.
\end{definition}

When $R = \Z_2$, Definition \ref{def advantage} relates to the more popular notion of fraction of satisfied constraints by an assignment; indeed $\text{adv}_{\I}(\bm{x})$ is simply the fraction of satisfied minus the fraction of violated constraints. 

When $R=\Z_q$, the advantage takes into account the `extent' to which a constraint is violated. For $q>2$, if $\bm{x}$ violates $(\bm{a},b)$ (i.e. $b\ne\innerproduct{\bm{a}}{\bm{x}}$), $b-\innerproduct{\bm{a}}{\bm{x}}$ can take any value in $\Z_q^*$. For $i\in[m]$, let $b_i = \innerproduct{\bm{a}_i}{\bm{x}} + e_i$ (i.e. $e_i$ is the error term for $\bm{x}$ with respect to $i\textsuperscript{th}$ constraint in the linear equation version; w.l.o.g., $-q/2\leq e_i\leq q/2$). We can rewrite the advantage as follows:
\begin{align}
    \texttt{adv}_{\I}(\bm{x}) = \frac{1}{m}\sum_{i\in[m]}\cos{\left(\frac{2\pi\abs{e_i}}{q}\right)} = 1 - \frac{2}{m}\left(\sum_{i\in[m]}\sin^2{\left(\frac{\pi\abs{e_i}}{q}\right)}\right) \label{advantage interpretation}
\end{align}
The $\sin(\cdot)$ function increases in the range of $\abs{e_i}$, so in some sense, each term in the summand measures the `extent' to which $\bm{x}$ violates $(\bm{a}_i,b_i)$.\\

Now we will define random and planted versions of $\LIN{k}{R}$ and $\ALIN{k}{R}$. To do this for $\ALIN{k}{R}$ we first define a distribution for choosing left-hand sides, generalizing the distribution described in \cite{sparse_LWE}.

\begin{definition} \label{defn sparse coefficients} (analogous to Definition 4.1 in \cite{sparse_LWE})
For a dimension \( n \in \mathbb{N} \), a sparsity parameter \(k\), a ring $R$, we define the distribution \( D_{\coeff, n, k, R} \) that first samples a uniformly random set \( S \in \binom{[n]}{k} \). Then, the distribution samples a random vector as follows: Sample \( \bm{a} \in R^n \) so that each coordinate \( a_j \) is set to $0$ if \( j \notin S \). Otherwise, \( a_j \) is chosen at random from \(R\).
\end{definition}

\begin{obs}\label{obs:pt}
    While sampling a vector from the distribution \( D_{\coeff, n, k, R} \), the probability of sampling $\bm{0}\in R^n$ is $p_0 = 1/\abs{R}^k$.
\end{obs}

\begin{definition}
    (Random instance) A $\LIN{k}{R}$ (resp. $\ALIN{k}{R}$) instance $\I = \{(\bm{a}_i,b_i)_{i\in[m]}\}$ is said to be a random instance if $\bm{a}_i\xleftarrow{\text{iid}} \text{Unif}(R^n_k)$ (resp. $\bm{a}_i\xleftarrow{\text{iid}} D_{\coeff, n, k, R}$) and $b_i \xleftarrow{\text{iid}} \text{Unif}(R)$.
\end{definition}

\begin{definition}
    (Planted instance) A $\LIN{k}{R}$ (resp. $\ALIN{k}{R}$) instance $\I = \{(\bm{a}_i,b_i)_{i\in[m]}\}$ is said to be a planted instance with respect to secret vector $\bm{s}\in R^n$ and noise distribution $\Psi$ on $R$ if $\bm{a}_i\xleftarrow{\text{iid}} \text{Unif}(R^n_k)$ (resp. $\bm{a}_i\xleftarrow{\text{iid}} D_{\coeff, n, k, R}$) and for each $i\in[m]$, $b_i = \innerproduct{\bm{a_i}}{\bm{s}}+e_i$ where $e_i\xleftarrow{iid} \Psi$.
\end{definition}

\subsubsection{Sparse LWE and Sparse LPN.} 

We first recall the discrete Gaussian distribution on the integers and $\Z_q$.

\begin{definition}\label{discrete gaussian on integers}
The discrete Gaussian distribution on the integers with width $r$, denoted $\Psi_{\Z,r}$, is the distribution having probability mass at $x\in\Z$ proportional to $\exp\left(-\pi x^2/r^2\right)$.
\end{definition}

\begin{definition}\label{discrete gaussian on Zq}
The discrete Gaussian distribution on $\Z_q$ with width $r$, denoted $\Psi_{\Z_q,r}$ is the distribution having probability mass at $j\in\Z_q$ proportional to $\sum\limits_{x\in q\Z+j}e^{-\pi x^2/r^2}$.
\end{definition}

Let us state a couple of observations related to the discrete Gaussian, which shall be useful later.
\begin{obs}\label{obs:psi0 discrete gaussian}
    For $1<r<q$, we have $\Psi_{\Z_q,r}(0) = \Theta(1/r)$ (see Appendix~\ref{app:psi0 for DG}).
\end{obs}

\begin{obs}\label{obs: rho for discrete gaussian}
    Let $\chi:\Z_q\to\mathbb{C}$ be the canonical character of $\Z_q$. For $1<r<q$, we have $\E[\chi(\Psi_{\Z_q,r})] \geq e^{-\pi r^2/q^2}$ (see Appendix~\ref{app:rho for DG}).
\end{obs}

We now recall the $k$-sparse LWE problem from \cite{sparse_LWE}.

\begin{definition} \label{defn sparse LWE} (Definition 4.3 in \cite{sparse_LWE}) 
For an integer dimension \( n \), sparsity parameter \( k \), sample complexity \( m \), modulus \( q \), and noise parameter \( r \), we define the \(k\mbox{-}\text{LWE}_{n,m,q,r}\) (resp. \(k\mbox{-}\text{ALWE}_{n,m,q,r}\)) problem as a distinguishing problem that is given $m$ samples from either of two $\LIN{k}{\Z_q}$ distributions, and is required to guess the distribution. The first distribution is a planted $\LIN{k}{\Z_q}$ (resp. $\ALIN{k}{\Z_q}$) distribution with noise distribution $\Psi_{\Z_q,r}$ and the second distribution is a random $\LIN{k}{\Z_q}$ (resp. $\ALIN{k}{\Z_q}$) distribution.
\end{definition}

The standard decisional LWE problem is identical to \(k\mbox{-}\text{ALWE}_{n,m,q,r}\)
problem for $k=n$.

We also define a sparse version of the Learning Parity with Noise(LPN) problem for higher modulus, which (the non-sparse version) was defined in \cite{LPN_for_higher_fields}.

\begin{definition} \label{defn sparse LPN} 
For an integer dimension \( n \), sparsity parameter \( k \), sample complexity \( m \), modulus \( q \), and noise parameter \(\mu\), we define the \(k\mbox{-}\text{LPN}_{n,m,q,\mu}\) (resp. \(k\mbox{-}\text{ALPN}_{n,m,q,\mu}\)) problem as a distinguishing problem that is given $m$ samples from either of two $\LIN{k}{\Z_q}$ distributions, and is required to guess the distribution. The first distribution is a planted $\LIN{k}{\Z_q}$ (resp. $\ALIN{k}{\Z_q}$) distribution with noise distribution $\zeta_{\Z_q,\mu}$ on $\Z_q$ where $\zeta_{\Z_q,\mu}(0) = 1-\mu$, $\forall j\in\Z_q\setminus\{0\}\ \zeta_{\Z_q,\mu}(j) = \mu/(q-1)$. The second distribution is a random $\LIN{k}{\Z_q}$ (resp. $\ALIN{k}{\Z_q}$) distribution.
\end{definition}

Let us compute $\E[\chi(\zeta_{\Z_q,\mu})]$ which will be useful for us.
\begin{obs}\label{obs: rho for lpn distribution}
  Let $\chi:\Z_q\to\mathbb{C}$ be the canonical character of $\Z_q$. Then $\E[\chi(\zeta_{\Z_q,\mu})] = 1-\frac{q\mu}{(q-1)}$.
\end{obs}
\begin{proof}
    \begin{align*}
        \E[\chi(\zeta_{\Z_q,\mu})] = \sum_{j\in\Z_q}\zeta_{\Z_q,\mu}(j)\chi(j)\ = 1-\mu + \frac{\mu}{q-1}\sum_{j\in\Z_q}\chi(j)\ = 1-\mu -\frac{\mu}{(q-1)}
    \end{align*}
\end{proof}

\section{A Kikuchi Graph Construction}\label{sec:Kikuchi graph}
In this section, our objective is to distinguish a random instance of $\LIN{k}{R}$ from a planted such instance. 

We will develop the insight of viewing a $k$-sparse $R$-linear equation as the edges of a graph in a large vertex space. Given a $\LIN{k}{R}$ instance $\I$, we wish to produce a directed graph, namely $\K_l(\I)$\ (for some parameter $l$). This technique was introduced for the $k$-XOR problem by Wein. et. al. \cite{kikuchi_ref1} (and implicitly by Hastings \cite{kikuchi_ref2}). Having a graph representative of the instance allows us to use spectral properties of the graph for the distinguishing problem. 

For a head start, here is some intuition on how this is done: Let us consider a transformed space whose dimensions are indexed by $R^n_{k/2}$ (assume for simplicity for now that $k$ is even). Then an equation over indeterminates $\{x_i\}_{i\in[n]}$
$$C: {a_{i_1}}x_{i_1} + \ldots + {a_{i_k}}x_{i_k} = b_i$$ can be converted to an equation over indeterminates $\{y_{\bm{u}}\}_{\bm{u}\in R^{n}_{k/2}}$
$$C': y_{\bm{v}} - y_{\bm{u}} = b_i$$ where $\bm{u},\bm{v}$ satisfy
$$u_{j} = \begin{cases}
    -a_{j} &\text{ for $j=i_1,\ldots,i_{k/2}$ } \\
    0&\text{ otherwise}
\end{cases}\ \ \ v_{j} = \begin{cases}
    a_{j} &\text{ for $j=i_{k/2+1},\ldots,i_{k}$ } \\
    0&\text{ otherwise}
\end{cases}$$ The rationale behind this conversion is this: For $\bm{s}\in R^n$, consider $\bm{s}^{\odot k/2}\in R^{R^n_{k/2}}$. As mentioned in Section~\ref{sec:preliminaries}, this  is defined as follows: $s^{\odot k/2}_{\bm{u}} := \sum_{j\in[n]}u_js_j$. Then $\bm{s}$ satisfies equation $C$ iff $\bm{s}^{\odot k/2}$ satisfies $C'$. 

Now we can define a graph over the vertex set $R^n_{k/2}$, where $C'$ can be viewed as an edge $\bm{u}\rightarrow \bm{v}$ with a sign that can be a function of $b_i$. We will define this sign to be $\chi(b_i)$; this will make the graph's adjacency matrix complex valued and will enable us to utilize the well-studied spectral theory of Hermitian matrices.

The above choice of non-zero indices for $\bm{u},\bm{v}$ is only one of $\binom{k}{k/2}$ possible choices; we can choose any of the $k/2$ indices for $\bm{u}$. Furthermore, note that the following generalizations can also be made:
\begin{itemize}
    \item We need not restrict ourselves to $(k/2)$ non-zero indices, we can have $l$ non-zero indices in $\bm{u},\bm{v}$ for an appropriate $l\geq k/2$.
    \item We also have some freedom to choose the values $u_j, v_j$, for instance we could have defined $u_{i_1} = 1-a_{i_1}, v_{i_1} = 1$ and still maintained that $\bm{s}$ satisfies $C$ iff $\bm{s}^{\odot k/2}$ satisfies $C'$. 
\end{itemize}

\begin{construction}\label{kLIN to 2LIN}
Let $\I$ be a $\LIN{k}{R}$ instance, and $\ceil{k/2}\leq l\leq n$ be a parameter. $\K_l(\I)$ is a directed graph over vertices indexed by $R^n_l$. For every constraint $(\bm{a},b) \in \I$
, $\K_l(\I)$ has edges $\bm{u}\xrightarrow{\chi(b)} \bm{v}$ and $\bm{v}\xrightarrow{\chi(-b)} \bm{u}$ for all $\bm{u},\bm{v}\in R^n_l$  which satisfy $\bm{v}-\bm{u} = \bm{a}$.
\end{construction}

\begin{remark}\label{rem:comparing construction with KM25}
    Our Kikuchi construction differs from that of \cite{kocurek2025spectral} (their case of $k$ even, $R$ field) in two aspects. Firstly, they had edges $\bm{u}\xrightarrow{\chi(\beta b)} \bm{v}$ for $\bm{v}-\bm{u}=\beta\bm{a}$ for all $\beta\in R\setminus\{0\}$. Secondly, they imposed an additional restriction on these edges that the symmetric difference of supports of $\bm{u}$ and $\bm{v}$ should be equal to the support of $\bm{a}$. When we compare our results with theirs in Subsection~\ref{subsec:compare with KM25}, we elaborate on how differing in the first aspect allows us to  distinguish random instances from planted ones for a large class of noise distributions, and on how differing in the second aspect allows us to have a unified and simplified construction for odd $k$ and $R$ being any ring when $\abs{R}>2$.
\end{remark}

For a constraint $(\bm{a},b)\in\I$, how many corresponding edges does $\K_l(\I)$ have? Before venturing to compute this, let us rewrite Construction~\ref{kLIN to 2LIN} in a way that is more amenable to computing this. 

\begin{construction}\label{kLIN to 2LIN alternate}
(Equivalent to Construction~\ref{kLIN to 2LIN}) For every constraint $(\bm{a},b) \in \I$, $\K_l(\I)$ has all edges $(\bm{u},\bm{v},\chi(b))$, $(\bm{v},\bm{u},\chi(-b))$ which satisfy the following:
\begin{enumerate}
    \item For every $j\in[n]$ such that $a_j=0$, exactly one of the following holds:
    \begin{enumerate}
        \item $u_j = v_j = 0$. 
        \item $u_j = v_j \ne 0$.  
    \end{enumerate}
    \item For every $j\in[n]$ such that $a_j\in R\setminus\{0\}$, exactly one of the following holds: 
    \begin{enumerate}
        \item $u_j = -a_j, v_j = 0$. 
        \item $u_j = 0, v_j = a_j$.  
        \item $u_j,v_j\in R\setminus\{0\}$ such that $v_j - u_j = a_j$.
    \end{enumerate}
\end{enumerate}    
\end{construction}

We will now compute $\theta^{(n,l,k)}$ which is the number of edges added in $\K_l(\I)$ corresponding to the constraint $(\bm{a},b)\in\I$.

Let $r_{1a}, r_{1b}, r_{2a}, r_{2b}, r_{2c}$ denote the number of indices on which conditions 1(a), 1(b), 2(a), 2(b), 2(c) are applied respectively. It can be seen that the following must be satisfied:
\begin{align}
    &r_{1a} + r_{1b} + r_{2a} + r_{2b} + r_{2c} = n\ \ \ (\text{total indices}) \label{total indices} \\
    &r_{1b} + r_{2a} + r_{2c} = l\ \ \ (\text{non-zero entries of }\bm{u}) \label{non-zero entries of u}\\
    &r_{1b} + r_{2b} + r_{2c} = l\ \ \ (\text{non-zero entries of }\bm{v}) \label{non-zero entries of v}\\
    &r_{2a} + r_{2b} + r_{2c} = k\ \ \  (\text{non-zero entries of }\bm{a})\label{non-zero entries of a}
\end{align}

If $r_{2c} = r \in \{0,1,\ldots,k\}$, we obtain $r_{2a}=r_{2b}=(k-r)/2$ from \eqref{non-zero entries of u}, \eqref{non-zero entries of v}, \eqref{non-zero entries of a}, $r_{1b} = l-(k+r)/2$ from \eqref{non-zero entries of u}, $r_{1a} = n-l-(k-r)/2$ from \eqref{total indices}. For indices $j$ of condition 1(b), we have $(\abs{R}-1)=\abs{R}-1$ choices of $(u_j,v_j)$, and for indices $j$ of condition 2(c), we have $\abs{R}-2$ choices of $(u_j,v_j)$.
 
Thus we get the number of constraints for each original constraint as: \footnote{in this equation we use the convention that if $s$ is not a non-negative integer or $s>t$, then $\binom{t}{s}=0$}
\begin{align}
    \theta^{(n,l,k)} = 2\sum_{r=0}^t\binom{k}{r}(\abs{R}-2)^r\binom{k-r}{(k-r)/2}\binom{n-k}{l-(k+r)/2}(\abs{R}-1)^{l-(k+r)/2} \label{theta}
\end{align}

\begin{example}
    Let us consider a couple of toy examples to get a better idea on the construction. Let $R=\Z_5$, dimension $n=4$ and $l=3$.

    \begin{enumerate}
        \item Consider $\bm{a}=(1,2,0,0)\in R^n_k$ for $k=2$. Two vectors $\bm{u},\bm{v}\in R^n_l$ corresponding to $\bm{a}$ (satisfying $\bm{v}-\bm{u}=\bm{a}$) are:
        \begin{enumerate}
            \item $\bm{u}=(-1,0,3,4), \bm{v}=(0,2,3,4)$. This corresponds to the case $r_{2c}=0$.
            \item $\bm{u}=(2,2,1,0), \bm{v}=(3,4,1,0)$. This corresponds to the case $r_{2c}=2$.
        \end{enumerate}
        \item Consider $\bm{a}=(1,2,3,0)\in R^n_{k}$ for $k=3$. Two vectors $\bm{u},\bm{v}\in R^n_l$ corresponding to $\bm{a}$ (satisfying $\bm{v}-\bm{u}=\bm{a}$) are:
        \begin{enumerate}
            \item $\bm{u}=(2,0,-3,4), \bm{v}=(3,2,0,4)$. This corresponds to the case $r_{2c}=1$.
            \item $\bm{u}=(1,1,1,0), \bm{v}=(2,3,4,0)$. This corresponds to the case $r_{2c}=3$.
        \end{enumerate}
    \end{enumerate}

\end{example}

Note that in \eqref{theta}, only the terms where $r$ has same parity as $k$ are positive. 

We can obtain two lower bounds for $\theta^{(n,l,k)}$. One of them will come from the summand of $r=0$ or $r=1$, depending on whether $k$ is even or odd, respectively.
 
\begin{align}
    \text{Assuming $\ceil{k/2}\leq l\leq n-\floor{k/2}$, }\ \ \theta^{(n,l,k)} \geq \frac{2^{k-1}}{(k+1)}\binom{n-k}{l-\ceil{k/2}}(\abs{R}-1)^{l-\floor{k/2}} \label{theta lower bound_r=0}
\end{align}

The other bound will come from the $r=k$ summand term.

\begin{align}
    \text{Assuming $k\leq l\leq n$, }\ \ \theta^{(n,l,k)} \geq 2^{1-k}\binom{n-k}{l-k}(\abs{R}-1)^{l} \label{theta lower bound r=k}
\end{align}

Depending on whether $\abs{R}$ is small or large, we will use these bounds for our analysis. 

\section{Distinguishing by spectral method over Kikuchi Graphs}\label{sec:Kikuchi spectral}

The Kikuchi Graph $\K_l(\I)$ can be naturally associated with its adjacency matrix (for simplicity, we use the same notation to denote the adjacency matrix)\footnote{if there are parallel edges, we add their signs in the corresponding matrix entry}. By Construction~\ref{kLIN to 2LIN}, for every edge $\bm{u}\xrightarrow{\chi(b)}\bm{v}$, we have the edge $\bm{v}\xrightarrow{\chi(-b)}\bm{u}$ present in the graph. Therefore $\K_l(\I)$ is Hermitian. 

Our certificate to determine whether $\I$ is planted or random will be the \emph{spectral norm} of $\K_l(\I)$. To give some intuition for this, let us relate the optimal advantage of $\I$ with the spectral norm of $\K_l(\I)$.
\begin{proposition}\label{prop:justifying Kikuchi construction}
    Let $\I = \{(\bm{a}_i,b_i)\}_{i\in[m]}$ be a $\LIN{k}{R}$ instance. For every assignment $\bm{x}\in R^n$, define $\bm{y}\in\mathbb{C}^{R^n_{l}}$ as $y_{\bm{u}} := \chi\left(-\sum_{j\in[n]}u_jx_j\right)$. Then $\text{adv}_{\I}(\bm{x}) = \frac{1}{m\theta^{(n,l,k)}}\bm{y}^{\dagger} \K_l(\I)\bm{y}$. Further, $\texttt{adv}(\I) \leq \frac{N}{m\theta^{(n,l,k)}}\norm{\K_l(\I)}$ for $N = \binom{n}{l}(\abs{R}-1)^l$.
\end{proposition}
\begin{proof} 
    \begin{align*}
        \bm{y}^{\dagger} \K_l(\I)\bm{y} &= \sum_{\bm{u},\bm{v}\in R^n_l}y_{\bm{u}}^*\cdot\left(\K_l(\I)\right)_{\bm{u}\bm{v}}\cdot y_{\bm{v}} \\
        &= \sum_{i\in[m]}\sum_{\substack{\bm{v}-\bm{u}=\bm{a}_i\\\bm{u}-\bm{v}=\bm{a}_i}}y_{\bm{u}}^*\cdot\left(\K_l(\I)\right)_{\bm{u}\bm{v}}\cdot y_{\bm{v}} \\
        &= \sum_{i\in[m]}\sum_{\bm{v}-\bm{u}=\bm{a}_i}\chi\left(\sum_{j\in[n]}u_jx_j\right)\cdot\chi\left(b_i\right)\cdot \chi\left(-\sum_{j\in[n]}v_jx_j\right) +\\ &\ \ \ \sum_{i\in[m]}\sum_{\bm{u}-\bm{v}=\bm{a}_i}\chi\left(\sum_{j\in[n]}u_jx_j\right)\cdot\chi\left(-b_i\right)\cdot \chi\left(-\sum_{j\in[n]}v_jx_j\right)  \\
        &= \sum_{i\in[m]}\left(\sum_{\bm{v}-\bm{u}=\bm{a}_i}\chi\left(b_i - \innerproduct{\bm{v}-\bm{u}}{\bm{x}}\right) + \sum_{\bm{u}-\bm{v}=\bm{a}_i}\chi\left(-b_i+ \innerproduct{\bm{u}-\bm{v}}{\bm{x}}\right)\right) \\
        &= \sum_{i\in[m]}\frac{\theta^{(n,l,k)}}{2}\cdotp\left(\chi\left(b_i-\innerproduct{\bm{a}_i}{\bm{x}}\right)+\chi\left(-b_i+\innerproduct{\bm{a}_i}{\bm{x}}\right)\right)\ \ (\text{each $(\bm{a}_i,b_i)$ had }\theta^{(n,l,k)}/2\text{ edge pairs})\\
        &= \theta^{(n,l,k)}\sum_{i\in[m]}\text{Re}\left(\chi\left(b_i-\innerproduct{\bm{a}_i}{\bm{x}}\right)\right)\\
        &= m\theta^{(n,l,k)}\texttt{adv}_{\I}(\bm{x})
    \end{align*}
    Further, note that for every $\bm{y}\in\mathbb{C}^{R^n_l}$ with $ \abs{y_{\bm{u}}}=1\ \forall\ \bm{u}\in R^n_l$ we have $\norm{\bm{y}}=\sqrt{N}$ and thus 
    \begin{align*}
        \abs{\bm{y}^{\dagger} \K_l(\I)\bm{y}} = N\cdot\abs{\left(\frac{\bm{y}^{\dagger}}{\norm{\bm{y}^{\dagger}}}\right)\K_l(\I)\left(\frac{\bm{y}}{\norm{\bm{y}}}\right)} \leq N\norm{\K_l(\I)}
    \end{align*}
    and the second part follows.
\end{proof}

Proposition~\ref{prop:justifying Kikuchi construction} gives some evidence that estimating the spectral norm of $\norm{\K_l(\I)}$ may be helpful for distinguishing a random vs planted instance. One consequence of this proposition is that we can hope $\norm{\K_l(\I)}$ to be high when $\I$ is planted. We will show that this is indeed the case with high probability. Moreover, to separate the random instance using this quantity, we also need to show that $\norm{\K_l(\I)}$ is low with high probability when $\I$ is random. The rest of the argument in this section is about proving tail bounds for $\norm{\K_l(\I)}$ in both cases.

\subsection{Independent Summands of $\K_l(\I)$}

Firstly for a $\LIN{k}{R}$ instance $\I = \{(\bm{a}_i,b_i)\}_{i\in[m]}$ we write $\K_l(\I)$ as a sum of independent Hermitian matrices:
\begin{align}
    \K_l(\I) &= \sum_{i\in[m]}\bm{K}_i \\
    \bm{K}_i(\bm{u},\bm{v}) &:= \begin{cases}
    \chi(b_i) & \text{if }\bm{v}-\bm{u} = \bm{a}_i \\
    \chi(-b_i) & \text{if }\bm{u}-\bm{v} = \bm{a}_i \\
    0 & \text{otherwise}
\end{cases} \label{eqn:Xi}
\end{align}

By definition, $\bm{K}_i$ are Hermitian, and independent as the constraints in $\I$ are independent. 

We now prepare to establish some key properties of $\bm{K}_i$. Observe that the graph corresponding to $\bm{K}_i$ is a collection of disjoint \emph{chains} (see illustration in Figure~\ref{fig:chain}). This is because the graph corresponding to $\bm{K}_i$ never contains constraints $(\bm{u},\bm{v},\chi(b_i))$ and $(\bm{u}',\bm{v}',\chi(b_i))$ with  ($\bm{u} = \bm{u}'\ \land\ \bm{v} \neq \bm{v}'$), or ($\bm{u} \neq \bm{u}'\ \land\ \bm{v} = \bm{v}'$). 

\begin{figure}[h]
    \centering
\[
    \begin{tikzcd}
    \dots 
    %\bm{u}_1
    \arrow[r, "\chi(b_i)", shift left=1ex] 
    & 
    \bm{u}
    \arrow[l, "\chi(-b_i)", shift left=1ex] 
    \arrow[r, "\chi(b_i)", shift left=1ex] 
    & 
    \bm{v} 
    \arrow[l, "\chi(-b_i)", shift left=1ex] 
    \arrow[r, "\chi(b_i)", shift left=1ex]
    & 
    \dots 
    \arrow[l, "\chi(-b_i)", shift left=1ex]
\end{tikzcd}
\]
    \caption{A chain in the graph of $\bm{K}_i$}
    \label{fig:chain}
\end{figure}

\begin{obs}\label{obs:Ki degree 2}
    The graph corresponding to $\bm{K}_i$ has maximum out-degree 2.
\end{obs}
\begin{proof}
$\bm{K}_i$ never contains constraints $(\bm{u},\bm{v},\chi(b_i))$ and $(\bm{u}',\bm{v}',\chi(b_i))$ with  ($\bm{u} = \bm{u}'\ \land\ \bm{v} \neq \bm{v}'$), or ($\bm{u} \neq \bm{u}'\ \land\ \bm{v} = \bm{v}'$). So, for a vertex $\bm{u}$, there are only 4 possible edges incident to $\bm{u}$: $$(\bm{u},\bm{v},\chi(b_i)), (\bm{v},\bm{u},\chi(-b_i)), (\bm{v}',\bm{u},\chi(b_i)), (\bm{u},\bm{v}',\chi(-b_i))$$ In particular, $\bm{u}$ has an out-degree of at most 2. 
\end{proof}

Let us prove an upper bound on the spectral norm of the adjacency matrix of a signed graph. This is a generalization of the well known spectral norm bound of an unsigned graph.

\begin{lemma}\label{spectral norm bounded by max degree}
    Let $\G$ be a directed graph (possibly with multiple parallel edges and self loops) where each edge is of the form $(I,J,\beta)$ with $\beta\in\mathbb{C}, \abs{\beta}\leq 1$. Then $\norm{\G} \leq D(\G)$. 
\end{lemma}

\begin{proof}
    Let $\bm{\phi}$ be an eigenvector corresponding to the eigenvalue of $\G$ whose absolute value is $\norm{\G}$. Let $I$ be a vertex in $\G$ such that $\abs{\phi_I} \geq \abs{\phi_J}$ for all vertices $J$ of $\G$. Since $\bm{\phi}\ne\bm{0}$, we have $\abs{\phi_I}\ne0$. Now
    \begin{align*}
        \norm{\G} &= \frac{\abs{(\G\phi)_I}}{\abs{\phi_I}} \leq \frac{\sum_{(I,J,\beta)\in E(\G)}\abs{\beta\cdot\phi_J}}{\abs{\phi_I}} \\
        &\leq \frac{\sum_{(I,J,\beta)\in E(\G)}{\abs{\phi_J}}}{\abs{\phi_I}} \leq \text{out-deg}(I) \leq D(\G)
    \end{align*}
\end{proof}

The following upper bound for $\norm{\bm{K}_i}$ follows as a corollary of Observation~\ref{obs:Ki degree 2} and Lemma~\ref{spectral norm bounded by max degree}. 

\begin{corollary}\label{Xi norm bound}
    Let $\I = \{(\bm{a}_i,b_i)\}_{i\in[m]}$ be any $\LIN{k}{R}$ instance, and $\bm{K}_i$ be as defined in \eqref{eqn:Xi}. Then $\norm{\bm{K}_i}\leq 2$.
\end{corollary}

Next, we state a result regarding sum of squares of $\{\bm{K}_i\}_{i\in[m]}$. Looking ahead, this will be useful to compute the variance statistic corresponding to $\K_l(\I)$. The key observation is that $\bm{K}_i^2$ can be viewed as an adjacency matrix of a graph that captures ``2-hops'' of the original graph of $\bm{K}_i$. 

\begin{lemma}\label{variance bound}
    Let $\I = \{(\bm{a}_i,b_i)\}_{i\in[m]}$ be any $\LIN{k}{R}$ instance, and $\bm{K}_i$ be as defined in \eqref{eqn:Xi}. Then $\norm{\sum_{i\in[m]}\bm{K}_i^2} \leq 2D(\K_l(\I))$, $\norm{\sum_{i\in[m]}\E\left[\bm{K}_i^2\right]} \leq 2D(\K_l(\I))$ (where the expectation is over the randomness of $\{b_i\}_{i}$)\footnote{$D(\K_l(\I))$ does not depend on $\{b_i\}_{i}$ as the latter only contribute to signs of the edges}. 
\end{lemma}
\begin{proof}
In the graph of $\bm{K}_i$, consider any edge pair
\[
    \begin{tikzcd}
    \dots\  
    \bm{u} 
    \arrow[r, "\chi(b_i)", shift left=1ex] 
    & 
    \bm{v} 
    \arrow[l, "\chi(-b_i)", shift left=1ex] 
    \dots 
\end{tikzcd}
\]
which is part of a chain. By Observation~\ref{obs:Ki degree 2}, each vertex in the graph of $\bm{K}_i$ has out-degree at most 2, so the edge $(\bm{u},\bm{v},\chi(b_i))$ can contribute to at most 2 outgoing edges of $\bm{u}$ in the graph corresponding to $\bm{K}_i^2$ (One is $\bm{u}\xrightarrow{\chi(b_i)}\bm{v}\xrightarrow{\chi(-b_i)}\bm{u}$ and the other is possibly $\bm{u}\xrightarrow{\chi(b_i)}\bm{v}\xrightarrow{\chi(b_i)}\bm{w}$). 

Let $\G$ be the graph corresponding to $\sum_{i\in[m]}\bm{K}_i^2$. We then have $\text{out-deg}_{\G}(\bm{u})\leq 2\cdotp\text{out-deg}_{\K_l(\I)}(\bm{u})$ for each vertex $\bm{u}$. By applying Lemma~\ref{spectral norm bounded by max degree} on $\G$, we obtain $\norm{\sum_{i\in[m]}\bm{K}_i^2} \leq 2D(\K_l(\I))$.

Further, for any edge $\bm{u}\xrightarrow{\beta}\bm{v}$ of $\bm{K}_i^2$, $\abs{\E[\beta]}\leq \E[\abs{\beta}]=1$ by Jensen's inequality. Thus the graph $\E[\G] = \sum_{i\in[m]}\E\left[\bm{K}_i^2\right]$ also satisfies the premise of Lemma~\ref{spectral norm bounded by max degree} with $\text{out-deg}_{\E\left[\G\right]}(\bm{u}) = \text{out-deg}_{\G}(\bm{u})\leq 2\cdotp\text{out-deg}_{\K_l(\I)}(\bm{u})$ for each vertex $\bm{u}$, and the result follows.
\end{proof}

Let us now obtain an upper bound on $D(\K_l(\I))$, for which we use the distribution of the `left-hand sides', i.e. the fact that $\bm{a}_i\xleftarrow{\text{iid}}\text{Unif}(R^n_k)$. Recall that this is the setting in both the random and planted instances. 

\begin{lemma}\label{app:expected degree in one step}
    Let $\I = \{(\bm{a}_i,b_i)\}_{i\in[m]}$ be a $\LIN{k}{R}$ instance where $\bm{a}_i\xleftarrow{\text{iid}}\text{Unif}(R^n_k)$. Let $N = \binom{n}{l}(\abs{R}-1)^l$, and $\theta^{(n,l,k)}$ be as defined in \eqref{theta}. For any $\bm{u}\in R^n_l$, let $d_{\bm{u}}^{(i)}$ be the random variable denoting the number of edges (at most $2$) corresponding to $(\bm{a}_i,b_i)$ outgoing from $\bm{u}$ in $\K_l(\I)$. Then $\mathbb{E}\left[d_{\bm{u}}^{(i)}\right] = \frac{\theta^{(n,l,k)}}{N}$. 
\end{lemma}
\begin{proof}
    $\bm{a}_i$ is equally likely to be any vector in $R^n_k$. By symmetry, for every $\bm{u}'\in R^n_l$, there is an equal number of vertex pairs $(\bm{u}',\bm{v})$ such that $\bm{v}-\bm{u}'\in R^n_k$. Since there are (deterministic) $\theta^{(n,l,k)}$ edges added in $\K_l(\I)$ for $(\bm{a}_i,b_i)$, we get from the above two observations that $\E\left[d_{\bm{u}}^{(i)}\right] = \frac{\theta^{(n,l,k)}}{N}$. 
\end{proof}

\begin{proposition}\label{upper bound on max degree}
    Let $\I = \{(\bm{a}_i,b_i)\}_{i\in[m]}$ be a $\LIN{k}{R}$ instance where $\bm{a}_i\xleftarrow{\text{iid}}\text{Unif}(R^n_k)$. Let $N = \binom{n}{l}(\abs{R}-1)^l$, $\Delta = \frac{m\theta^{(n,l,k)}}{N}$. Then for any $\delta > 0$, we have $D(\K_l(\I)) \;\leq\; (1+\delta)\Delta$
except with probability at most $N \cdot \exp\!\left(-\,\frac{\delta^{2}\Delta}{2(2+\delta)}\, \right)$.
\end{proposition}
\begin{proof}
As a consequence of Lemma~\ref{app:expected degree in one step}, $\outdeg(\bm{u}) = \sum_{i\in[m]}d_{\bm{u}}^{(i)}$ satisfies $$\E[\outdeg(\bm{u})] = \frac{m\theta^{(n,l,k)}}{N} = \Delta$$ Further, $\outdeg(\bm{u})$ is a sum of independent random variables $\left\{d_{\bm{u}}^{(i)}\right\}_{i\in[m]}$ supported on $\{0,1,2\}$, and \\$D(\K_l(\I)) = \max\limits_{\bm{u}\in R^n_l}{\outdeg(\bm{u})}$. Thus we can apply a Chernoff bound $\left(\text{Theorem~\ref{thm:chernoff} on }\left\{d_{\bm{u}}^{(i)}/2\right\}_{i\in[m]}\right)$ followed by a union bound to prove the proposition.    
\end{proof}

\subsection{An upper bound for the random instance}

We will show that when $\I$ is a random $\LIN{k}{R}$ instance, an upper bound on $\norm{\K_l(\I)}$ holds with high probability. To establish this, let us apply the Matrix Bernstein inequality to our case.  

We start with analysing the summands $\bm{K}_i$'s further when $\I$ is a random instance.

\begin{lemma}\label{X_i expected value}
    Let $\I$ be a random $\LIN{k}{R}$ instance, and $\bm{K}_i$ be as defined in \eqref{eqn:Xi}. Then $\E[\bm{K}_i] = \bm{O}$.
\end{lemma}
\begin{proof}
    An entry of $\bm{K}_i$ can be $\chi(b_i),\chi(-b_i)$ or $0$. Now $\E[\chi(b_i)] = \frac{1}{\abs{R}}\sum_{a\in R}\chi(a) = 0$ from Observation~\ref{obs:sum of chi}. Similarly $\E[\chi(-b_i)] = 0$ and the result follows.
\end{proof}

With this, the antecedent of Theorem \ref{matrix bernstein} is satisfied (taking $\bm{X}_i = \bm{K}_i$) with $\norm{\bm{K}_i} \leq 2$ (from Corollary~\ref{Xi norm bound}) for all $i\in[m]$. Also, from Lemma \ref{variance bound}, we have
\begin{align}\label{X_i sigma2}
    \sigma^2 = \left\| \sum_{i\in[m]} \mathbb{E}\left[\bm{K}_i^{2}\right] \right\| \leq 2 D(\K_l(\I))
\end{align}

Now we choose an appropriate value of $t$ to limit the probability in Theorem \ref{matrix bernstein} to be an inverse polynomial of $N = \abs{R^n_l} = \binom{n}{l}(\abs{R}-1)^l$, say $N^{-\epsilon}$ for some choice of $\epsilon > 0$. Thus we must satisfy
\begin{align}
    &N \cdot \exp\!\left( \frac{-t^{2}/2}{2D(\K_l(\I)) + 2t/3} \right) \leq N^{-\epsilon} \label{quadratic in t}
\end{align}

By simplifying \eqref{quadratic in t} we have the following upper bound for the maximum eigenvalue:

\begin{corollary}\label{upper bound on lambda_max wrt D(G)}
Let $\I = \{(\bm{a}_i,b_i)\}_{i\in[m]}$ be a random $\LIN{k}{R}$ instance of dimension $n$, let $N = \binom{n}{l}(\abs{R}-1)^l$. Then for any $\epsilon>0$, except with probability at most $N^{-\epsilon}$, we have
\begin{align}
    \norm{\K_l(\I)} \;\leq\; 
\frac{2\sqrt{(1+\epsilon)\ln N}}{3}\Big(\sqrt{(1+\epsilon)\ln N} + \sqrt{9D(\K_l(\I))+(1+\epsilon)\ln N}\Big) \label{lambda_max bound wrt D(G)}
\end{align}
\end{corollary}

When the left-hand sides $\bm{a}_i\xleftarrow{\text{iid}}\text{Unif}(R^n_k)$ (which is the case for a random instance), Proposition \ref{upper bound on max degree} gives us that $D(\K_l(\I))$ is upper bounded by $(1+\delta)\Delta$ with high probability, where $\Delta = \frac{m\theta^{(n,l,k)}}{N}$ is the average out-degree of $\K_l(\I)$, and $\delta>0$ is a constant. We can now combine Corollary \ref{upper bound on lambda_max wrt D(G)} and Proposition \ref{upper bound on max degree} to get an upper bound on the spectral norm with respect to $\Delta$. We would like to choose $\Delta$ (which in turn means choosing the number of constraints $m$ in $\I$) so that we get $\norm{\K_l(\I)} \leq \delta \Delta$. Observe that in \eqref{lambda_max bound wrt D(G)} we can impose
\[
\Delta \geq \frac{16\delta + 12}{3\delta^2} (1+\epsilon)\ln N
\]
to get this desired bound for all $\delta>0$. Also, the probability of failure in Proposition \ref{upper bound on max degree} is upper bounded by $N^{-\epsilon}$ for this choice.

\begin{theorem}\label{delta' lower bound random case}
Let $\I = \{(\bm{a}_i,b_i)\}_{i\in[m]}$ be a random $\LIN{k}{R}$ instance of dimension $n$, let $N = \binom{n}{l}(\abs{R}-1)^l$. Suppose $\Delta := \frac{m\theta^{(n,l,k)}}{N} \geq \frac{16\delta + 12}{3\delta^2}(1+\epsilon)\ln N$ for some $\delta,\epsilon > 0$. Then we have $\norm{\K_l(\I)} \leq \delta\Delta$ except with probability at most $2N^{-\epsilon}$.
\end{theorem}

Finally let us restate this result in terms of the parameters of $\I$. 

Here we will use the lower bounds of $\theta^{(n,l,k)}$ from \eqref{theta lower bound_r=0}, \eqref{theta lower bound r=k} depending on the value of $\abs{R}$ and $l$.

If $\abs{R} = O(n/l)$, then \eqref{theta lower bound_r=0} will give a better bound than \eqref{theta lower bound r=k}. On the other hand, if $\abs{R} = \omega(n/l)$, \eqref{theta lower bound r=k} gives a better bound.

\begin{corollary}\label{upper bound for random instance small q}
Let $k,n,l$ be positive integers such that $\ceil{k/2}\leq l\leq n-\floor{k/2}$. Let $0<\delta<1$. Consider a $\LIN{k}{R}$ instance $\I$ of dimension $n$ and $m$ constraints, \\ $m \geq \frac{l\ln((\abs{R}-1)n)}{\delta^2}\cdotp(\abs{R}-1)^{\floor{k/2}}\left(\frac{\alpha n}{l}\right)^{\ceil{k/2}}$ for an absolute constant $\alpha > 0$. Let $\Delta := \frac{m\theta^{(n,l,k)}}{\binom{n}{l}(\abs{R}-1)^l}$. Then if $\I$ is random, we get $\norm{\K_l(\I)} \leq \delta\Delta$ except with probability at most $n^{-\Omega(l)}$.
\end{corollary}

\begin{corollary}\label{upper bound for random instance large q}
Let $k,n,l$ be positive integers such that $k\leq l\leq n$. Let $0<\delta<1$. Consider a $\LIN{k}{R}$ instance $\I$ of dimension $n$ and $m$ constraints, 
$m \geq \frac{l\ln((\abs{R}-1)n)}{\delta^2}\cdotp\left(\frac{\alpha n}{l}\right)^{k}$ for an absolute constant $\alpha > 0$. Let $\Delta := \frac{m\theta^{(n,l,k)}}{\binom{n}{l}(\abs{R}-1)^l}$. Then if $\I$ is random, we get $\norm{\K_l(\I)} \leq \delta\Delta$ except with probability at most $n^{-\Omega(l)}$.
\end{corollary}

\subsection{A lower bound for the planted instance}

Recall that the only difference in a planted instance from a random instance of $\LIN{k}{R}$ is that the `right-hand side' $b_i\in R$ of a constraint $(\bm{a}_i,b_i)$ is correlated with $\bm{a}_i$.  

We will impose a mild assumption on the noise distribution $\Psi$ for the planted instance: the quantity $\rho := \E[\chi(\Psi)]\in\R\setminus\{0\}$.\footnote{the noise for sparse LWE, sparse LPN and other nice distributions shall satisfy this assumption.} 

The heart of our analysis will still be the Matrix Bernstein inequality, but with a slightly different set of matrices. This is because Lemma \ref{X_i expected value} ($\E[\bm{K}_i] = \bm{O}$) will no longer hold for the same definition of $\bm{K}_i$ when $\I$ is planted. Let us calculate $\E[\bm{K}_i]$. 

\begin{notation}\label{not:A_z,ai} 
Let $\bm{a}\in R^n_{k}$ and $\bm{s}\in R^n$. We define a square matrix $A_{\bm{s},\bm{a}}$ indexed by $R^n_l$ as follows:
\begin{align}
    A_{\bm{s},\bm{a}}(\bm{u},\bm{v}) &:= \begin{cases}
    \chi(\innerproduct{\bm{a}}{\bm{s}}) & \text{if }\bm{v}-\bm{u} = \bm{a} \\
    \chi(-\innerproduct{\bm{a}}{\bm{s}}) & \text{if }\bm{u}-\bm{v} = \bm{a} \\
    0 & \text{otherwise}
\end{cases}    
\end{align}
\end{notation}

\begin{lemma}\label{X_i expected value for planted case}
    Let $\I = \{(\bm{a}_i,b_i)\}_{i\in[m]}$ be a planted $\LIN{k}{R}$ instance with respect to secret vector $\bm{s}\in R^n$ and noise distribution $\Psi$ on $R$ with $\rho = \E[\chi(\Psi)]\in\R$. Let $\bm{K}_i$'s be the same as defined in \eqref{eqn:Xi}. Then $\E[\bm{K}_i] = \rho A_{\bm{s},\bm{a}_i}$.
\end{lemma}
\begin{proof}
    This follows from the fact that $b_i = \innerproduct{\bm{a_i}}{\bm{s}}+e_i$, $\chi(b_i) = \chi(\innerproduct{\bm{a_i}}{\bm{s}})\cdot\chi(e_i)$ and $\E[\chi(e_i)] = \E[\chi(-e_i)] = \rho$.
\end{proof}

Notice that $\sum_{i\in[m]}A_{\bm{s},\bm{a}_i}$ can be viewed as a signed adjacency matrix of an appropriate graph. This graph has the same edges as in $\K_l(\I)$ but with different signs. Analogous to the spectral norm of an adjacency matrix being lower bounded by the average degree of the graph, we will set a lower bound for the spectral norm of $\sum_{i\in[m]}A_{\bm{s},\bm{a}_i}$ which will be useful later.
\begin{proposition}\label{eigenvalue lower bound for signed adjacency matrix}
    Let $\I = \{(\bm{a}_i,b_i)\}_{i\in[m]}$ be a planted $\LIN{k}{R}$ instance with respect to secret vector $\bm{s}\in R^n$ and noise distribution $\Psi$ on $R$. Let $N := \abs{R^n_l} = \binom{n}{l}(\abs{R}-1)^l$, $\Delta := \frac{m\theta^{(n,l,k)}}{N}$. Then for any $\delta>0$ we have  $\norm{\sum_{i\in[m]}A_{\bm{s},\bm{a}_i}} \geq \Delta$.
\end{proposition}
\begin{proof}
    \begin{align}
        \norm{\sum_{i\in[m]}A_{\bm{s},\bm{a}_i}} = \sup_{\bm{x}\in\mathbb{C}^{R^n_l}\setminus\{\bm{0}\}}{\frac{\bm{x}^\dagger\left(\sum_{i\in[m]}A_{\bm{s},\bm{a}_i}\right)\bm{x}}{\bm{x}^\dagger\bm{x}}} \label{eqn:A_z quadratic form}
    \end{align}

    Let us consider the quadratic form in \eqref{eqn:A_z quadratic form} for a specific $\bm{x}$. Define $\bm{z}\in \mathbb{C}^{R^n_l}$ as follows: For $\bm{u}\in R^n_l$, $z_{\bm{u}} := \chi(-\innerproduct{\bm{u}}{\bm{s}})$. It can be seen that for $\bm{u},\bm{v}\in R^n_l$ satisfying $\bm{v}-\bm{u} = \bm{a}$, we have $\left(z_{\bm{u}}\right)^{-1}z_{\bm{v}} = \chi(-\innerproduct{\bm{a}}{\bm{s}})$.

    With this, we have
    \begin{align*}
        \norm{\sum_{i\in[m]}A_{\bm{s},\bm{a}_i}} &\geq \frac{\left(\sum\limits_{i\in[m]}\bm{z}^\dagger\cdotp A_{\bm{s},\bm{a}_i}\cdotp\bm{z}\right)}{\bm{z}^\dagger\bm{z}} \\
        &= \frac{\sum_{i\in[m]}\sum_{\substack{(\bm{u},\bm{v})}}{z_{\bm{u}}^{-1}\cdotp A_{\bm{s},\bm{a}_i}(\bm{u},\bm{v})}\cdotp z_{\bm{v}}}{N} \\
        &= \frac{\sum_{i\in[m]}(\text{\# of non-zero entries in $A_{\bm{s},\bm{a}_i}$})}{N} \\
        &= \Delta
    \end{align*}
    where the last inequality holds because the graph of $\sum_{i\in[m]}A_{\bm{s},\bm{a}_i}$ has the same edges as in $\K_l(\I)$.
\end{proof}

Let us now define $\widetilde{\bm{K}_i} := \bm{K}_i - \E[\bm{K}_i] = \bm{K}_i - \rho A_{\bm{s},\bm{a}_i}$. By definition, $\E[\widetilde{\bm{K}_i}] = \bm{O}$. Let us also obtain a bound on its spectral norm.

\begin{lemma}\label{X_i_tilde spectral norm}
    For $\widetilde{\bm{K}_i}$ defined as above, $\norm{\widetilde{\bm{K}_i}} \leq 2(1+\abs{\rho})$.
\end{lemma}
\begin{proof}
    $\norm{\widetilde{\bm{K}_i}} \leq \norm{\bm{K}_i} + \abs{\rho}\norm{A_{\bm{s},\bm{a}_i}} \leq 2(1+\abs{\rho})$ where the second inequality can be obtained by an analogous argument for $A_{\bm{s},\bm{a}_i}$ as for $\bm{K}_i$ in Corollary~\ref{Xi norm bound}.
\end{proof}

To apply the Matrix Bernstein inequality, we are left to compute 
$$\sigma^2 := \left\| \sum_{i\in[m]} \mathbb{E}\left(\widetilde{\bm{K}_i}^{2}\right) \right\|$$
\begin{align}
    \sigma^2 &= \left\| \sum_{i\in[m]} \mathbb{E}\left(\left(\bm{K}_i - \rho A_{\bm{s},\bm{a}_i}\right)^2\right) \right\| \ \ = \left\| \sum_{i\in[m]} \left(\mathbb{E}\left(\bm{K}_i^2\right) - \rho^2 A_{\bm{s},\bm{a}_i}^2\right) \right\| \\
    &\leq \left\| \sum_{i\in[m]} \mathbb{E}\left(\bm{K}_i^2\right)\right\| + \abs{\rho}^2 \left\|\sum_{i\in[m]}A_{\bm{s},\bm{a}_i}^2 \right\| \\
    &\leq 2(1+\abs{\rho}^2)D(\K_l(\I)) \ \ (\text{using Lemma~\ref{variance bound} and same idea for }A_{\bm{s},\bm{a}_i}^2)
\end{align}

Like the random case, we need to choose an appropriate value of $t$ to limit the probability in Theorem \ref{matrix bernstein} to be essentially an inverse polynomial of $N = \abs{R^n_l} = \binom{n}{l}(\abs{R}-1)^l$, say $N^{-\epsilon}$ for some choice of $\epsilon > 0$. Thus we must satisfy
\begin{align}
    &N \cdot \exp\!\left( \frac{-t^{2}/2}{2\left(1+\abs{\rho}^2\right)D(\K_l(\I)) + 2(1+\abs{\rho})t/3} \right) \leq N^{-\epsilon} \label{quadratic in t for planted case}
\end{align}

It can be seen that \eqref{quadratic in t for planted case} is satisfied for (denote $Z = (1+\epsilon)\ln N$ in below equation)
\begin{align}
    t \geq \frac{2}{3}\left( (1+\abs{\rho})Z
    + \sqrt{
         \big((1+\abs{\rho})Z\big)^2
        + 9\left(1+\abs{\rho}^2\right) D(\K_l(\I)) Z
    }\right) \label{t threshold for planted case}
\end{align}

Like the random case, the next step would have been to bound $D(\K_l(\I))$. Fortunately, we can again assume $D(\K_l(\I))\leq (1+\gamma)\Delta$ from Proposition~\ref{upper bound on max degree}; even when $\I$ is planted, since  $\bm{a}_i\xleftarrow{\text{iid}}\text{Unif}(R^n_k)$. 

So, let us now proceed to get an upper bound on $\norm{\sum_{i\in[m]}\widetilde{\bm{K}_i}}$ with respect to $\Delta$. Like before, we would like to choose $\Delta$ so that we get $\norm{\sum_{i\in[m]}\widetilde{\bm{K}_i}} \leq \gamma \Delta$. Observe that in equation \eqref{t threshold for planted case} we can impose
\[
\Delta \geq 
\frac{12(1+\abs{\rho}^2)(1+\gamma)+4\gamma(1+\abs{\rho})}{3\gamma^2}(1+\epsilon)\ln N
\]
to get this desired bound for all $\gamma>0$. Again, the failure probability of Proposition~\ref{upper bound on max degree} is bounded by $N^{-\epsilon}$ for this choice. 

We now have $\norm{\sum_{i\in[m]}\widetilde{\bm{K}_i}} \leq \gamma\Delta$. How do we proceed to get a lower bound for $\norm{\sum_{i\in[m]}{\bm{K}_i}} = \norm{\K_l(\I)}$? 
\begin{align*}
    \norm{\K_l(\I)} &= \norm{\sum_{i\in[m]}{\bm{K}_i}} \ \geq \norm{\sum_{i\in[m]}{\E[\bm{K}_i}]} - \norm{\sum_{i\in[m]}\widetilde{\bm{K}_i}} \\
    &\geq \abs{\rho}\norm{\sum_{i\in[m]}{A_{\bm{s},\bm{a}_i}}} - \gamma\Delta \ \ (\text{from Lemma \ref{X_i expected value for planted case}}) \\
    &\geq \abs{\rho}\Delta - \gamma\Delta \ \ (\text{from Proposition \ref{eigenvalue lower bound for signed adjacency matrix}})
\end{align*}

\begin{theorem}\label{delta' lower bound planted case}
Let $\I = \{(\bm{a}_i,b_i)\}_{i\in[m]}$ be a planted $\LIN{k}{R}$ instance of dimension $n$, noise distribution $\Psi$ on $R$ with $\rho = \E[\chi(\Psi)]$. Let $N = \binom{n}{l}(\abs{R}-1)^l$. Suppose
$$\Delta := \frac{m\theta^{(n,l,k)}}{N} \geq \frac{12(1+\abs{\rho}^2)(1+\gamma)+4\gamma(1+\abs{\rho})}{3\gamma^2}(1+\epsilon)\ln N$$ Then we have $\norm{\K_l(\I)} \geq (\abs{\rho} - \gamma)\Delta$ except with probability at most $2N^{-\epsilon}$.
\end{theorem}

Finally let us restate this result in terms of the parameters of $\I$. Once again we will use the lower bounds of $\theta^{(n,l,k)}$ from \eqref{theta lower bound_r=0}, \eqref{theta lower bound r=k}; the better result can be used depending on whether $\abs{R} = O(n/l)$ or $\abs{R} = \omega(n/l)$.

\begin{corollary}\label{lower bound for planted instance small q}
Let $k,n,l$ be positive integers such that $\ceil{k/2}\leq l\leq n-\floor{k/2}$. Let $0<\gamma<1$. Consider a $\LIN{k}{R}$ instance $\I$ of dimension $n$ and $m$ constraints, $m \geq \frac{l\ln((\abs{R}-1)n)}{\gamma^2}\cdotp(\abs{R}-1)^{\floor{k/2}}\left(\frac{\alpha n}{l}\right)^{\ceil{k/2}}$ where $\alpha > 0$ is an absolute constant. Let $\Delta := \frac{m\theta^{(n,l,k)}}{\binom{n}{l}(\abs{R}-1)^l}$. Then if $\I$ is planted with noise distribution $\Psi$ on $R$ such that $\rho := \E[\chi(\Psi)]\in\R\setminus\{0\}$, we get that $\norm{\K_l(\I)} \geq (\abs{\rho} - \gamma)\Delta$ except with probability at most $n^{-\Omega(l)}$.
\end{corollary}

\begin{corollary}\label{lower bound for planted instance large q}
Let $k,n,l$ be positive integers such that $k\leq l\leq n$. Let $0<\gamma<1$. Consider a $\LIN{k}{R}$ instance $\I$ of dimension $n$ and $m$ constraints, $m \geq \frac{l\ln((\abs{R}-1)n)}{\gamma^2}\cdotp\left(\frac{\alpha n}{l}\right)^{k}$ where $\alpha > 0$ is an absolute constant. Let $\Delta := \frac{m\theta^{(n,l,k)}}{\binom{n}{l}(\abs{R}-1)^l}$. Then if $\I$ is planted with noise distribution $\Psi$ on $R$ such that $\rho := \E[\chi(\Psi)]\in\R\setminus\{0\}$, we get that $\norm{\K_l(\I)} \geq (\abs{\rho} - \gamma)\Delta$ except with probability at most $n^{-\Omega(l)}$.
\end{corollary}

\subsection{Distinguishing between a planted and a random instance}\label{subsec:distinguish planted vs random Kikuchi}
Corollaries \ref{upper bound for random instance small q}, \ref{upper bound for random instance large q} and Corollaries \ref{lower bound for planted instance small q}, \ref{lower bound for planted instance large q} together suggest that we can solve the decisional $\LIN{k}{R}$ problem with high probability if we had the ability to approximate the maximum eigenvalue of $\K_l(\I)$. For this, we need to ensure that the thresholds established in these corollaries are separated. In particular, we must have $(\abs{\rho} - \gamma)\Delta > \delta\Delta$, which can be easily achieved by setting, say, $\gamma = \delta = \abs{\rho}/4$. 

We can then use a standard eigenvalue estimation technique, say the Power Iteration method, to approximate $\norm{\K_l(\I)}$. Note that $\K_l(\I)$ is a $N\times N$ matrix for $N = \binom{n}{l}(\abs{R}-1)^l$. Moreover, it is a sparse matrix since each row has at most $D(\K_l(\I)) = O(\Delta)$ non-zero entries, where $\Delta = O(\ln N)$ suffices for our purpose. Thus, the Power Iteration takes $\widetilde{O}(N) = \widetilde{O}\left(\binom{n}{l}(\abs{R}-1)^l\right)$ time to converge.

\begin{corollary}\label{solving kLIN small q}
Let $k,n,l$ be positive integers such that $\ceil{k/2}\leq l\leq n-\floor{k/2}$. Consider a $\LIN{k}{R}$ instance $\I$ of dimension $n$ and $m$ constraints which is either random or is planted with noise distribution $\Psi$ on $R$ with $\rho = \E[\chi(\Psi)]\in\R\setminus\{0\}$. Let $m \geq \frac{l\ln((\abs{R}-1)n)}{\abs{\rho}^2}\cdotp(\abs{R}-1)^{\floor{k/2}}\left(\frac{\alpha n}{l}\right)^{\ceil{k/2}}$  where $\alpha > 0$ is an absolute constant. Then we can determine the distribution of $\I$ in $\widetilde{O}\left(\binom{n}{l}(\abs{R}-1)^l\right)$ time except with probability at most $n^{-\Omega(l)}$.
\end{corollary}

\begin{corollary}\label{solving kLIN large q}
Let $k,n,l$ be positive integers such that $k\leq l\leq n$. Consider a $\LIN{k}{R}$ instance $\I$ of dimension $n$ and $m$ constraints which is either random or is planted with noise distribution $\Psi$ on $R$ with $\rho = \E[\chi(\Psi)]\in\R\setminus\{0\}$. Let $m \geq \frac{l\ln((\abs{R}-1)n)}{\abs{\rho}^2}\cdotp\left(\frac{\alpha n}{l}\right)^{k}$ where $\alpha > 0$ is an absolute constant. Then we can determine the distribution of $\I$ in $\widetilde{O}\left(\binom{n}{l}(\abs{R}-1)^l\right)$ time except with probability at most $n^{-\Omega(l)}$.
\end{corollary}

An important aspect of this distinguisher is that it works for any (nice enough) noise distribution of the planted instance. The sample complexity required by the algorithm is inversely proportional to $\abs{\rho}^2$ where $\rho = \E[\chi(\Psi)]$. The sample and time complexities do not depend on any other property of the noise distribution. Also, as long as $\abs{\rho}$ is sufficiently large (say at least inverse polynomial in $n$), it won't affect the sample complexity by a lot. As immediate consequences of Corollaries \ref{solving kLIN small q}, \ref{solving kLIN large q} we have algorithms for solving $k$-sparse LWE and sparse LPN over modulus $q$. 

For sparse LWE, we can compute $1/\rho^2 \leq e^{2\pi r^2/q^2}$ (Observation~\ref{app:psi0 for DG}). Since the width $r$ for the LWE distribution is chosen to be less than $q$, this dependency will only contribute a constant term.

\begin{theorem}\label{solving sparse LWE}
    Let $n,m,q,r$ be positive integers, $r<q/2$. Let $\alpha > 0$ be an absolute constant. There exists an algorithm that solves the decisional \(k\mbox{-}\text{LWE}_{n,m,q,r}\) problem such that:
    \begin{itemize}
        \item Let $l\in\N$ be a parameter such that $\ceil{k/2}\leq l\leq n-\floor{k/2}$. Then provided that the number of samples $m\geq l\ln(qn)\cdotp (q-1)^{\floor{k/2}}\left(\frac{\alpha n}{l}\right)^{\ceil{k/2}}$, the algorithm takes time $\widetilde{O}\left(\binom{n}{l}(q-1)^l\right)$ and is correct, except with probability at most $n^{-\Omega(l)}$.
        \item Let $l\in\N$ be a parameter such that $k\leq l\leq n$. Then provided that the number of samples $m\geq l\ln(qn)\cdotp \left(\frac{\alpha n}{l}\right)^{k}$, the algorithm takes time $\widetilde{O}\left(\binom{n}{l}(q-1)^l\right)$ and is correct except with probability at most $n^{-\Omega(l)}$.
    \end{itemize}
\end{theorem}

For sparse LPN, $\rho = 1 - \frac{q\mu}{q-1}$ from Definition \ref{defn sparse LPN}. (For $\mu\rightarrow \frac{q-1}{q}$, the planted distribution approaches the random distribution.) 

\begin{theorem}\label{solving sparse LPN}
    Let $n,m,q$ be positive integers, $0\leq \mu<(q-1)/q$. Let $\alpha > 0$ be an absolute constant. There exists an algorithm that solves the decisional \(k\mbox{-}\text{LPN}_{n,m,q,\mu}\) problem such that:
    \begin{itemize}
        \item Let $l\in\N$ be a parameter such that $\ceil{k/2}\leq l\leq n-\floor{k/2}$. Then provided that the number of samples $m\geq \left(1-\frac{q\mu}{q-1}\right)^{-2}(l\ln(qn))\cdotp(q-1)^{\floor{k/2}}\left(\frac{\alpha n}{l}\right)^{\ceil{k/2}}$, the algorithm takes time $\widetilde{O}\left(\binom{n}{l}(\abs{R}-1)^l\right)$ time and is correct except with probability at most $n^{-\Omega(l)}$.
        \item  Let $l\in\N$ be a parameter such that $k\leq l\leq n$. Then provided that the number of samples $m\geq \left(1-\frac{q\mu}{q-1}\right)^{-2}(l\ln(qn))\cdotp\left(\frac{\alpha n}{l}\right)^{k}$, the algorithm takes time $\widetilde{O}\left(\binom{n}{l}(\abs{R}-1)^l\right)$ time and is correct except with probability at most $n^{-\Omega(l)}$.
    \end{itemize}
\end{theorem}

\subsection{A quantum speedup} 
In essence, to distinguish a random $\LIN{k}{R}$ instance from a planted one, the Kikuchi method performs estimation of the spectral norm of a sparse Hermitian matrix in a space of dimension $N = \binom{n}{l}(\abs{R}-1)^l$. For any known classical method, such as the power iteration, eigenvalue estimation requires matrix multiplications, so it is unlikely for the time complexity to go below $\widetilde{\Omega}(N)$. 

However, with quantum algorithmic techniques, we can do better. We can use the same idea as given by Schmidhuber et. al. in \cite{quartic_speedup} for $R = \Z_2$: Using a \emph{guiding state} created from the $\LIN{k}{R}$ instance, we are able to determine whether the eigenspace corresponding to eigenvalues $\geq \lambda$ (for some threshold $\lambda > 0$) is non-empty. This enjoys a speedup on Power iteration which can be from quadratic to quartic depending on the choice of parameter $l$. We elaborate on this intuition in Appendix~\ref{app:quantum speedup}. For even more details, the reader may refer to Appendix C of \cite{quartic_speedup}.

\section{A Simple Algorithm}\label{sec:simple algo}

In this section, we describe another algorithm for distinguishing whether a $\LIN{k}{R}$ instance is random or planted. This algorithm is simpler than the spectral method and outperforms the latter when $\abs{R}$ is large, specifically $\abs{R} = \omega(n/l)$; like the spectral method, $l$ here is a parameter used for sample-time complexity tradeoff. This algorithm in essence, is the dense-minor attack proposed in \cite{sparse_LWE}. It is also similar to the simple refutation algorithm proposed in Section 8 of \cite{kocurek2025spectral}, but is more efficient since we process only one dense minor instead of all $\binom{n}{l}$ dense minors\ \footnote{for the refutation task, the latter was required}. We describe this in Algorithm~\ref{alg:dense minor}.

\begin{algorithm}
\caption{\textsc{Dense Minor}}
\label{alg:dense minor}
\begin{algorithmic}[1]
\REQUIRE $\LIN{k}{R}$ instance $\I$ which is either random or is planted with respect to distribution $\Psi$ on $R$ with $\rho := \E[\chi(\Psi)]\ne0$, tradeoff parameter $k\leq l\leq n$

\STATE $\I' := \{(\bm{a},b)\in\I : \bm{a}\text{ has non-zero entries only on the first $l$ coordinates}\}$
\STATE $\texttt{adv} := \max\left\{\abs{\frac{1}{\abs{\I'}}\sum\limits_{(\bm{a},b)\in\I'}\chi\left(b-\innerproduct{\bm{a}}{\bm{x}}\right)} : \bm{x}\in R^n \text{ has non-zero entries only on the first $l$ coordinates}\right\}$
\IF{$\texttt{adv}< \abs{\rho}/2$}{
    \RETURN \texttt{RANDOM}
}\ENDIF
\RETURN \texttt{PLANTED}
\end{algorithmic}
\end{algorithm}

Essentially, the algorithm distinguishes by (naively) computing the advantage for a sub-instance $\I'$ of $\I$. This is efficient enough for us as it only needs to check for assignments $\bm{x}$ having only first $l$ coordinates non-zero. In particular, Algorithm~\ref{alg:dense minor} takes time $O(mn + \abs{R}^l\abs{\I'}k\log\abs{R})$ time since this is the time taken by steps 1 and 2. 

Firstly we state that $\I'$ defined in step 1 of the algorithm has sufficiently many constraints. This can be shown via a routine calculation.
\begin{lemma}
    Let $\I$ be a $\LIN{k}{R}$ instance with dimension $n$ and $m$ constraints. Let $S\in\binom{[n]}{l}$ be a fixed set of indices, and let $\I_S := \{(\bm{a},b)\in\I : \bm{a}\text{ has non-zero entries only on indices of $S$}\}$. Suppose $m \geq \frac{l\ln((\abs{R}-1)n)}{\gamma^2}\cdotp\left(\frac{\alpha n}{l}\right)^{k}$ for some constants $\alpha,\gamma>0$. Then $\abs{\I_S}\geq \Omega(\gamma^{-2})l\log\left(\abs{R}n\right)$ except with probability $n^{-\Omega(l)}$.
\end{lemma}
\begin{proof}
    Since a left-hand side $\bm{a}$ is chosen uniformly and independently from $R^n_k$, the probability its support lies in $S$ is given by $p = \frac{\binom{l}{k}}{\binom{n}{k}} \geq \left(\frac{l}{en}\right)^k$. The expected number of constraints in $S$ are $mp \geq \frac{C^kl\ln((\abs{R}-1)n)}{\gamma^2}$ for some constant $C$. We can apply a Chernoff Bound (Theorem~\ref{thm:chernoff}):
    \begin{align*}
        \Pr\left[\abs{\I_S} < (1-\delta)\frac{C^kl\ln((\abs{R}-1)n)}{\gamma^2}\right] < e^{-\delta^2\frac{C^kl\ln((\abs{R}-1)n)}{2\gamma^2}} = n^{-\Omega(l)}
    \end{align*}
\end{proof}

We proceed by assuming $\I'$ has $m'\geq \frac{l\ln(\abs{R}n)}{\abs{\rho}^2}$ constraints, where $\rho := \E[\chi(\Psi)]\ne0$ for noise $\Psi$. For $\bm{x}\in R^n$, let us denote $\texttt{adv}_{\bm{x}} := \frac{1}{\abs{\I'}}\sum\limits_{(\bm{a},b)\in\I'}\chi\left(b-\innerproduct{\bm{a}}{\bm{x}}\right)$.

We first analyze the random case. If $\I$ is random, for every $(\bm{a},b)\in\I'$, $b$ is uniformly random in $R$, and so is $b-\innerproduct{\bm{a}}{\bm{x}}$ for an arbitrary $\bm{x}\in R^n$. From Observation~\ref{obs:sum of chi}, we have $\E\left[\chi(b-\innerproduct{\bm{a}}{\bm{x}})\right]=0$. So 
\begin{align}
    \forall\ \bm{x}\in R^n,\ &\E[\texttt{adv}_{\bm{x}}] = \frac{1}{m'}\sum\limits_{(\bm{a},b)\in\I'}\E\left[\chi(b-\innerproduct{\bm{a}}{\bm{x}})\right] = 0 \\
    \Rightarrow &\Pr\left[\abs{\texttt{adv}_{\bm{x}}} > \abs{\rho}/2\right] < 2e^{-m'\abs{\rho}^2/8} \label{eqn:hoeffding on random} \\
    \Rightarrow&\Pr\left[\texttt{adv} > \abs{\rho}/2\right] < 2\abs{R}^le^{-m'\abs{\rho}^2/8} \leq 2n^{-l/8} \label{eqn:adv random}
\end{align}
where the second inequality follows from applying Hoeffding inequality (Theorem~\ref{thm:Hoeffding}) over random variables $\{\chi(b-\innerproduct{\bm{a}}{\bm{x}})\}_{(\bm{a},b)\in\I'}$, and the last inequality follows by a union bound over all $\bm{x}\in R^n$ having non-zero entries only on first $l$ coordinates.

Now let us see what happens in the planted case. Suppose $\I$ is planted with respect to the secret vector $\bm{s}\in R^n$. Let $\bm{s}'\in R^n$ be defined as follows:
\begin{align*}
    s'_j = \begin{cases}
        s_j\ \text{if }j\leq l\\
        0\ \text{if }j>l
    \end{cases}
\end{align*}

Then for every $(\bm{a},b)\in\I'$, $b = \innerproduct{\bm{a}}{\bm{s}}+e = \innerproduct{\bm{a}}{\bm{s}'}+e$ such that $\E[\chi(e)]=\rho$. 

Thus
\begin{align}
    &\E[\texttt{adv}_{\bm{s}'}] = \frac{1}{m'}\sum\limits_{(\bm{a},b)\in\I'}\E\left[\chi(b-\innerproduct{\bm{a}}{\bm{s}'})\right] = \frac{1}{m'}\sum\limits_{(\bm{a},b)\in\I'}\E\left[\chi(e)\right] = \rho \\
    \Rightarrow &\Pr\left[\texttt{adv} < \abs{\rho}/2\right] \leq \Pr\left[\abs{\texttt{adv}_{\bm{s}'}} < \abs{\rho}/2\right] \\
    &\leq \Pr\left[\abs{\texttt{adv}_{\bm{s}'}-\rho} > \abs{\rho}/2\right] < 2e^{-m'\rho^2/8} \leq 2(\abs{R}n)^{-l/8}\ \ (\text{using Theorem~\ref{thm:Hoeffding}})
\end{align}

Hence we get the following result:

\begin{theorem}\label{thm:simple solving kLIN}
    Let $k,n,l$ be positive integers such that $k\leq l\leq n$. Consider a $\LIN{k}{R}$ instance $\I$ of dimension $n$ and $m$ constraints which is either random or is planted with noise distribution $\Psi$ on $R$ with $\rho := \E[\chi(\Psi)]$. Let $m \geq \frac{l\ln((\abs{R}-1)n)}{\abs{\rho}^2}\cdotp\left(\frac{\alpha n}{l}\right)^{k}$ where $\alpha > 0$ is an absolute constant. Then Algorithm~\ref{alg:dense minor} correctly distinguishes the distribution of $\I$ in $\widetilde{O}\left(\abs{\rho}^{-2}\left(\left(\frac{\alpha n}{l}\right)^{k+1} + \abs{R}^l\right)\right)$ time except with probability at most $n^{-\Omega(l)}$.
\end{theorem}

\begin{remark}
    By comparing Theorem~\ref{thm:simple solving kLIN} and Corollary~\ref{solving kLIN large q}, Algorithm~\ref{alg:dense minor} outperforms the spectral method over the Kikuchi graph when $\abs{R} = \omega(n/l)$. 

    So the advantage of spectral method over the Kikuchi graph is when $\abs{R}=O(n/l)$. For instance, if $\abs{R}=n^{\delta}$ for some $0\leq\delta<1$, then Algorithm~\ref{alg:dense minor} will require roughly $(n/l)^k$ samples to solve in $\left(n^{\delta l}\right)$ time\footnote{assuming $k\ll l$, ignoring $\rho$ and logarithmic terms for simplicity}.
    On the other hand, it can be seen from Corollary~\ref{solving kLIN small q} that   
    the spectral method requires roughly $\left(\frac{(1+\delta)\cdot n^{(1+\delta)k/2}}{\delta \cdot l^{k/2}}\right)$ samples to distinguish in the same time. 
\end{remark}

\section{Distinguishing random vs planted $\ALIN{k}{R}$ instances}\label{sec:alin}
Let us also analyze the distinguishing problem for $\ALIN{k}{R}$, i.e. when the left-hand sides are at most $k$-sparse and follow the distribution in Definition~\ref{defn sparse coefficients}. As it may seem intuitively, we will get substantially better time and sample complexities in this version of the problem compared to $\LIN{k}{R}$.

Firstly, it can be seen that Algorithm~\ref{alg:dense minor} will work just as well if $\I$ is a $\ALIN{k}{R}$ instance. Hence the problem enjoys the sample-time complexity tradeoff as laid out in Theorem~\ref{thm:simple solving kLIN}. 

\begin{theorem}\label{thm:simple solving ALIN large q}
    Let $k,n,l$ be positive integers such that $k\leq l\leq n$. Consider a $\ALIN{k}{R}$ instance $\I$ of dimension $n$ and $m$ constraints which is either random or is planted with noise distribution $\Psi$ on $R$ with $\rho = \E[\chi(\Psi)]\ne 0$. Let $m \geq \frac{l\ln((\abs{R}-1)n)}{\abs{\rho}^2}\cdotp\left(\frac{\alpha n}{l}\right)^{k}$ where $\alpha > 0$ is an absolute constant. Then Algorithm~\ref{alg:dense minor} correctly distinguishes the distribution of $\I$ in $\widetilde{O}\left(\abs{\rho}^{-2}\left(\left(\frac{\alpha n}{l}\right)^{k+1} + \abs{R}^l\right)\right)$ time except with probability at most $n^{-\Omega(l)}$.
\end{theorem}

We now describe another strategy. Recall from Observation~\ref{obs:pt} that the probability of sampling $\bm{0}\in R^n$ as a left-hand side is $p_0 = \abs{R}^{-k}$. What if we have enough samples $m$ in $\I$ so that there are $m'\geq l\ln n/\abs{\rho}^2$ constraints in $\I$ of the form $(\bm{0},b)$? Let us call this collection $\I'=\{(\bm{0},b)\}\subseteq\I$. Similar to Algorithm~\ref{alg:dense minor}, we can compute
\begin{align*}
    \texttt{adv} = \frac{1}{m'}\sum_{(\bm{0},b)\in\I'}\chi(b)
\end{align*}
If $\I$ is random, we have 
\begin{align}
    &\E[\texttt{adv}] = \frac{1}{m'}\sum\limits_{(\bm{0},b)\in\I'}\E\left[\chi(b)\right] = 0 \\
    \Rightarrow&\Pr\left[\abs{\texttt{adv}} > \abs{\rho}/2\right] < 2e^{-m'\abs{\rho}^2/8} \leq 2n^{-l/8}\ \ (\text{using Theorem~\ref{thm:Hoeffding}}) \label{eqn:adv random alin}
\end{align}
On the other hand, if $\I$ is planted, for $(\bm{0},b)\in\I'$ we have $b = e\xleftarrow{\text{iid}}\Psi$ such that $\E[\chi(e)]=\rho$.
\begin{align}
    &\E[\texttt{adv}] = \frac{1}{m'}\sum\limits_{(\bm{0},b)\in\I'}\E\left[\chi(b)\right] = \frac{1}{m'}\sum\limits_{(\bm{0},b)\in\I'}\E\left[\chi(e)\right] = \rho \\
    \Rightarrow &\Pr\left[\abs{\texttt{adv}} < \abs{\rho}/2\right] \leq \Pr\left[\abs{\texttt{adv}-\rho} > \abs{\rho}/2\right] < 2e^{-m'\rho^2/8} \leq 2n^{-l/8}\ \ (\text{using Theorem~\ref{thm:Hoeffding}})
\end{align}
Hence, like the dense minor idea, we can distinguish on the basis of whether $\abs{\texttt{adv}}\leq \abs{\rho}/2$. 

We know that the probability of a left-hand side being $\bm{0}$ is $p_0 = \abs{R}^{-k}$, so a standard application of Chernoff bound (Theorem~\ref{thm:chernoff}) gives that we need $m\geq 2\abs{R}^k\abs{\rho}^{-2}(l\ln n)$ samples to obtain the required number of $(\bm{0},b)$ samples ($m'\geq\abs{\rho}^{-2}l\ln n$) with probability at least $1-n^{-\Omega(l)}$. 

Finally, the time complexity of this strategy is $\widetilde{O}\left(\abs{\rho}^{-2}\abs{R}^k\right)$ as all we need to do is collect all samples with left-hand sides $\bm{0}$ and compute the sum of the corresponding (real parts of) right-hand sides. 

\begin{theorem}\label{thm:simple solving ALIN small q}
    Let $k,n,l$ be positive integers. Consider a $\ALIN{k}{R}$ instance $\I$ of dimension $n$ and $m$ constraints which is either random or is planted with noise distribution $\Psi$ on $R$ with $\rho := \E[\chi(\Psi)]\ne 0$. Let $m \geq 2\abs{R}^k\abs{\rho}^{-2}(l\ln n)$ where $\alpha > 0$ is an absolute constant. Then there exists an algorithm that correctly distinguishes the distribution of $\I$ in $\widetilde{O}\left(\abs{\rho}^{-2}\abs{R}^k\right)$ time except with probability at most $n^{-\Omega(l)}$.
\end{theorem}

Once again, we can choose the best algorithm depending on whether $\abs{R} = \omega(n/l)$ (Theorem~\ref{thm:simple solving ALIN large q}) or $\abs{R} = O(n/l)$ (Theorem~\ref{thm:simple solving ALIN small q}).

\section{Discussion}\label{sec:discussion}

\subsection{Comparison with \cite{kocurek2025spectral}}\label{subsec:compare with KM25}

Independent of our work, Kocurek and Manohar \cite{kocurek2025spectral} have studied the problem of refuting random and semirandom instances of $\LIN{k}{R}$\footnote{a semirandom instance is when the left-hand sides of a $\LIN{k}{R}$ instance can be arbitrarily chosen, but the right-hand sides are uniformly and independently chosen}. For this task, they define Kikuchi graphs for finite fields and abelian groups that are similar to our construction. Their sample-time tradeoffs are also similar to ours. However, our objective is different from the refutation task; ours is the task of distinguishing a random vs a planted instance. The reason why their methods cannot be extended to our problem is that the statistic they use, the {\em value} of an instance (maximum fraction of satisfiable constraints) doesn't allow us to distinguish the random instance from the planted instance when the noise in the planted distribution is too high. Our distinguishing statistic, the \emph{advantage} (Definition~\ref{def advantage}), doesn't suffer from this problem. As we have seen in Subsection~\ref{subsec:distinguish planted vs random Kikuchi}, the advantage can help us distinguish the two cases for any noise distribution satisfying a mild condition.

Moreover, they construct separate Kikuchi constructions depending on the sparsity $k$ being even or odd, and also depending on the ring being a field or non-field, with significantly more involved construction in the odd sparsity case. 

Our Kikuchi construction is able to mitigate both these issues; it is more amenable to the distinguishing task, and is a unified and simplified construction for all rings $R$ and both sparsity parities when $\abs{R}>2$. Further, we are able to construct a simpler distinguisher which has a significantly lower runtime than that of their simple refutation algorithm when the size of the ring $R$ is large. 

Let us elaborate on these aspects: 

\subsubsection{Distinguishing random vs planted instances for any planted noise} In the task of refutation of a $\LIN{k}{R}$ instance $\I$, the goal is to output a certificate that a given instance is highly unsatisfiable, meaning no assignment can satisfy more than a certain fraction of the constraints. To solve the task of refuting a random/semirandom $\LIN{k}{R}$ instance $\I$, the algorithm of \cite{kocurek2025spectral} returns a certificate $\algval(\I)\in[0,1]$ with the following guarantees: 
    \begin{enumerate}
        \item $\algval(\I)\geq\val(\I)$, where $\val(\I)$ is the maximum fraction of satisfied constraints by any assignment.
        \item If $\I$ is random, then $\algval(\I)\leq \frac{1}{\abs{R}}+\varepsilon$ for an arbitrary small $\varepsilon>0$ with high probability, which is close to the expected value $\left(\E\left[\texttt{val}(\I)\right] = \frac{1}{\abs{R}}\right)$. 
        \item The algorithm requires $m\geq \varepsilon^{-2}O(n)\left((\abs{R}-1)n/l\right)^{k/2-1}$ samples for the refutation task.
    \end{enumerate}

    To utilize their algorithm for the distinguishing task, a natural thought is to utilize Item 1 of above for the planted case, hoping that $\val(\I)$ will be high for the planted case. However, we elaborate below why this will not hold for noise distributions with low probability ($\sim 1/q$) of zero error, and so for those cases we won't be able to use their construction for the distinguishing task. We also elaborate how our construction mitigates this.
    
    Let us see what the quantity $\val(\I)$ looks like if $\I$ is a planted instance with respect to a secret assignment $\bm{s}\in R^n$ and noise distribution $\Psi$ on $R$. For any sample $(\bm{a},b)\in\I$, we have $b=\innerproduct{\bm{a}}{\bm{s}}+e$ for $e\xleftarrow{\text{iid}}\Psi$, so the probability that $\bm{s}$ satisfies the constraint $(\bm{a},b)$ is $\Psi(0)$. This in turn means that the expected number of satisfied constraints by $\bm{s}$ is $\Psi(0)$, and $\E[\val(\I)]\geq\Psi(0)$. 

    This will actually suffice to solve (distinguish) the sparse LPN problem over the ring $\Z_q$, as $\Psi(0) = 1-\mu$ (see Definition~\ref{defn sparse LPN}) which can be assumed to be non-negligibly larger than the expected value in the random case $1/q$ (say $\mu\geq\frac{q-1}{q}-\frac{1}{n^d}$ for some constant $d>0$). Due to this, the $\varepsilon$ parameter will contribute a factor of $n^{2d}$ in the sample complexity cost. Note that this matches our result for sparse LPN (Theorem~\ref{solving sparse LPN}).

    For sparse LWE, $\Psi = \Psi_{\Z_q,r}$ for which $\Psi_{\Z_q,r}(0) = O(1/r)$ when the width $r$ is super constant (from Observation~\ref{obs:psi0 discrete gaussian}). For instance, if $r=q/10$, we must have 
    $$\algval(\I)\leq 1/q+\varepsilon <  \Psi(0) = 10/q\ \Rightarrow \varepsilon < 9/q$$ In this case, $\varepsilon$ contributes a factor of $q^2$ to the sample complexity, making the total contribution of the modulus $q^{k/2+1}$. In contrast, our result (Theorem~\ref{solving sparse LWE}) has a factor of $q^{k/2}$ in the sample complexity.

    % \amitabha{The two paragraphs below should be moved up right after the definition of $\algval(\I)$. Instead of spending time going over the details first, we should crisply point out the difference first and then go for the details.}\shashwat{please check updated Section 8.1 and 8.1.1}

    But what if the noise distribution is slightly perturbed, so that the mass at $0\in\Z_q$ is $1/q$? Indeed, we can have a noise $\Psi'$ resembling the discrete gaussian, except that $\Psi'(0)=1/q$, and $\Psi'(j)=\Psi_{\Z_q,r}(j)+\Psi_{\Z_q,r}(0)/(q-1)-1/(q(q-1))$ for $j\in\Z_q\setminus\{0\}$. In this case, distinguishing on the basis of $\val(\I)$ is not possible. 

    On the other hand, our distinguisher examines a different quantity $\texttt{adv}(\I)$ (recall Definition~\ref{def advantage}) to separate the random and planted instances. In the random case, this quantity has expectation $0$. For the planted case we have $\E\left[\texttt{adv}(\I)\right] \geq \E\left[\text{adv}_{\I}(\bm{s})\right] = \E[\text{Re}\left(\chi(\Psi')\right)] =: \rho'$. To estimate $\rho'$, firstly note that $\rho :=  \E[\chi(\Psi_{\Z_q,r})] \geq e^{-\pi r^2/q^2} \geq e^{-\pi/4}$ when $r<q/2$ (Observation~\ref{obs: rho for discrete gaussian}). Then
    \begin{align*}
        \rho' := \E[\text{Re}\left(\chi(\Psi')\right)] &= \sum_{j\in\Z_q}e^{2\pi\iota j/q}\Psi'(j)\ \ \ (\text{note that $\Psi'(j)=\Psi'(-j)$}) \\
        &= 1/q + \sum_{j\in\Z_q\setminus\{0\}}e^{2\pi\iota j/q}\left(\Psi_{\Z_q,r}(j)+\Psi_{\Z_q,r}(0)/(q-1)-1/(q(q-1))\right) \\
        &= \rho - \left(\Psi_{\Z_q,r}(0)-1/q\right)\left(q/(q-1)\right) \\
        &\geq \rho - 2/r\ = \Omega(1) 
    \end{align*}
    Hence, in this case, $\rho'$ is still a constant, and the same sample complexity as for sparse LWE suffices for the distinguishing task. In fact, our method is \emph{oblivious} to the noise distribution $\Psi$, as long as it satisfies a mild condition ($\E[\text{Re}\left(\chi(\Psi)\right)]\in(0,1]$). 

    In Observation 3.6 of \cite{kocurek2025spectral}, Kocurek and Manohar relate the value of the instance with a quadratic form of their Kikuchi construction. One key aspect of their Kikuchi construction (for their case of even $k$, $R$ finite field) is that for a constraint $(\bm{a},b)\in\I$, they consider all edges $(\bm{u},\bm{v})\in R^n_l$ such that $\bm{v}-\bm{u} = \beta\bm{a}$ for all $\beta\in R\setminus\{0\}$. This is equivalent to adding constraints $(\beta\bm{a},\beta b)\in\I$ for all $\beta\in R\setminus\{0\}$ for each original constraints. 
    Informally, equation \eqref{eqn:informal contri of constraint to quadtratic form} below captures the contribution of this constraint to the quadratic form for some assignment $\bm{x}\in R^n$:
    \begin{align}
        \sum_{\beta\in R\setminus\{0\}}\chi\left(\beta b-\beta\innerproduct{\bm{a}}{\bm{x}}\right) = \sum_{\beta\in R\setminus\{0\}}\chi\left(\beta\left(b-\innerproduct{\bm{a}}{\bm{x}}\right)\right) = \begin{cases}
            (\abs{R}-1)\ \text{ if }b=\innerproduct{\bm{a}}{\bm{x}} \\
            -1\text{ if }b\ne\innerproduct{\bm{a}}{\bm{x}}
        \end{cases} \label{eqn:informal contri of constraint to quadtratic form}
    \end{align}

    So the extent of the noise (how far is $b$ from $\innerproduct{\bm{a}}{\bm{x}}$) is not captured in the quadratic form.

    On the other hand, we only keep edges of the form $\bm{v}-\bm{u} = \pm\bm{a}$. By this, we will achieve the appropriate contribution of this constraint in the advantage:

    \begin{align}
        \sum_{\beta\in \{\pm1\}}\chi\left(\beta b-\beta\innerproduct{\bm{a}}{\bm{x}}\right) = 2\text{Re}\left(\chi\left(b-\innerproduct{\bm{a}}{\bm{x}}\right)\right)
    \end{align}
    
    We show in Proposition~\ref{prop:justifying Kikuchi construction} how the quadratic form of our Kikuchi construction relates to our notion of advantage. 

\subsubsection{A unified construction} \cite{kocurek2025spectral} produce separate Kikuchi constructions and analyses depending on the sparsity parameter $k$ being odd or even, and also depending on the ring $R$ being a field or $\Z_{q_1}\times\ldots\times\Z_{q_t}$. 

\begin{itemize}
    \item Given a left-hand side $\bm{a}\in R^n_k$, their construction for even $k$ only considers edges $(\bm{u},\bm{v})\in R^n_l\times R^n_l$ wherein the symmetric difference of supports of $\bm{u}$ and $\bm{v}$ equal to the support of $\bm{a}$ (see Remark~\ref{rem:comparing construction with KM25}). Such pairs will not exist when $k$ is odd. So for odd $k$, they revert to a much more involved setup which includes partitioning $\I$ into multiple sub-instances, building a Kikuchi graph for each sub-instance (which differs from the even case), and performing a further processing step of deleting edges from high-degree vertices before taking the spectral norm. 

    In our construction, we do not restrict our edges to satisfy the above restriction for even $k$. This means that when $\abs{R}>2$, it is possible that for $\bm{v}-\bm{u}=\bm{a}$, we have $u_j,v_j,a_j\ne 0$ for some $j\in[n]$ (condition 2(c) of Construction~\ref{kLIN to 2LIN alternate}). This also allows valid pairs $(\bm{u},\bm{v})$ for $\bm{a}\in R^n_k$ when $k$ is odd and $\abs{R}>2$. 

    \item When $R$ is the abelian group $G = \Z_{q_1}\times\ldots\times\Z_{q_t}$, they reduce the instance to a quotient group $G/H$ which is either small, is of prime order (i.e. a field) or does not have small subgroups. This helps them to save a factor of $\abs{R}$ in their sample complexity dependence. Here too, their Kikuchi matrix (for even $k$ too) differs from the field case.

    For reasons stated above, our construction does not consider edges of the form $\bm{v}-\bm{u}=\beta\bm{a}$ for $\beta\in R\setminus\{0\}\setminus\{\pm1\}$, which helped them (when $R$ was a field) save a factor of $\abs{R}$ in their sample complexity. Thus using our construction for non-fields as well does not increase dependence on $\abs{R}$ in our sample complexity.
\end{itemize}

We also remark that our analysis works when $R$ is any ring apart from a finite field or $\Z_{q_1}\times\ldots\times\Z_{q_t}$, i.e. any arbitrary multiplication operation on the ring is allowed.

\subsubsection{Other Differences}

\begin{itemize}

    \item They use the Trace Moment method to bound the spectral norm of the Kikuchi matrix, as they deal with the semirandom instance, where the left-hand sides are arbitrary. Since we deal with instances where the left-hand sides are also random (in both random and planted instances), we use the Matrix Bernstein inequality whose analysis is simpler.

    \item They give a simpler algorithm for the refutation task which outperforms the Kikuchi method when $\abs{R}$ is large, specifically $\abs{R}=\omega(n/l)$. This algorithm takes same time as their Kikuchi method which is $(\abs{R}n)^{O(l)}$. This algorithm involves checking all assignments in $R^l$ for each $\binom{n}{l}$ $l$-sized subsets of $[n]$.
    
    We too give a simple algorithm in the large $\abs{R}$ case, which takes time $\widetilde{O}\left(\rho^{-2}\left(\left(\frac{\alpha n}{l}\right)^{k+1} + \abs{R}^l\right)\right)$. We are able to achieve this significant advantage in time, as processing a fixed $l$-sized subset suffices for the task of distinguishing a random vs planted instance.

\end{itemize}

\subsection{Comparison with \cite{sparse_LWE}}

Jain, Lin and Saha in \cite{sparse_LWE} introduced the sparse LWE problem over the ring $\Z_p$ for a prime $p$. There is an ambiguity in their definition, as they mention that each sample has \emph{exactly} $k$-non-zero entries, yet each non-zero entry is sampled uniformly from $\Z_p$ rather than $\Z_p^*$. In this work, we have defined and analyzed both versions ($\LIN{k}{R}$ and $\ALIN{k}{R}$). 

They attempt to claim hardness of sparse LWE, both by a reduction from standard LWE, and via a conjecture. We now discuss the implications of our algorithms on these hardness assumptions. 

\subsubsection{On hardness of \(k\mbox{-}\text{ALWE}_{n,m,p,r}\) assuming hardness of LWE}

Over the ring $R=\Z_p$, they show a reduction from standard LWE over dimension $k$ to $k$-sparse LWE over dimension $n\geq k$ (Theorem 5.1 of \cite{sparse_LWE}) for the at most $k$-non-zero entries version, i.e. \(k\mbox{-}\text{ALWE}_{n,m,p,r}\) as per our notation in Definition~\ref{defn sparse LWE}. The reduction preserves the number of samples $m$. It is a standard assumption in cryptography that solving standard LWE on dimension $k$ takes time $2^{\Omega(k)}$ time given arbitrary many samples (even as many as $p^{O(k)}$). Based on this assumption, their result establishes that with these many samples, solving \(k\mbox{-}\text{ALWE}_{n,m,p,r}\) also takes $2^{\Omega(k)}$ time. 

On the other hand, in Section~\ref{sec:alin} we show that solving \(k\mbox{-}\text{ALWE}_{n,m,p,r}\) is achievable in $\widetilde{O}\left(p^k\right)$ time given enough ($m=\widetilde{O}\left(p^k\right)$) samples. So when the number of samples is large and the modulus $p$ is small, our result is towards matching the hardness assumption for \(k\mbox{-}\text{ALWE}_{n,m,p,r}\) from their reduction.

\subsubsection{On conjectured hardness of \(k\mbox{-}\text{LWE}_{n,m,p,r}\), \(k\mbox{-}\text{ALWE}_{n,m,p,r}\)}

As an interesting case, let us analyze the sample complexity required by our dense minor algorithm (see Theorem~\ref{thm:simple solving kLIN}) if we allow a time of $2^{\delta n}$ for an arbitrarily small $\delta>0$. We have roughly $m\geq \left(\frac{n}{l}\right)^k$ samples required for a time of roughly $(n/l)^k + p^l$. We must have
\begin{align*}
    &(n/l)^k + p^l \approx 2^{\delta n} \ \ \ \Rightarrow l \approx \frac{\delta n}{\log p} \\
    \Rightarrow &m= O\left(\left(\log p\right)^{k}\right) \text{ samples are sufficient}
\end{align*}
In particular, if $k = O(\log n/\log\log p)$ then polynomially many samples suffice. If $k = O(\log n)$, then we require $(\log p)^{O(\log n)}$ samples which is polynomial for constant modulus $p$, and slightly more than polynomial for $p = O\left(\text{exp}(n)\right)$.

Conjecture 7.1 in \cite{sparse_LWE} suggests that \(k\mbox{-}\text{LWE}_{n,m,p,r}\) is as hard as standard LWE over $\Theta(n)$ dimension, $\text{poly}(n)$ samples and the same modulus and width parameters, when the sparsity $k = \Omega(\log n)$ and $m = \text{poly}(n)$. Assuming hardness of standard LWE, the algorithm from the Kikuchi method nearly closes the gap of this conjecture, either when the sparsity is slightly less, or the polynomial samples requirement is slightly relaxed. 

\printbibliography

\appendix
    \section{Facts on Discrete Gaussian}\label{app:Kikuchi graph}

    Recall the discrete Gaussian distribution $\Psi_{\Z_q,r}$ on $R=\Z_q$ with width parameter $r$ (Definition~\ref{discrete gaussian on Zq}). We will state a couple of facts regarding this distribution which we use. 

    Recall the Poisson Summation Formula for $f:\R\to\mathbb{R}$ where $\hat{f}$ denotes its Fourier transform:
    $$\sum_{x\in\mathbb{Z}} f(x) = \sum_{y\in\mathbb{Z}} \hat{f}(y)$$ 

    \subsection{Estimating $\Psi_{\Z_q,r}(0)$}\label{app:psi0 for DG}
    We have 
    $$\Psi_{\Z_q,r}(0)\ = \frac{\sum\limits_{x\in q\Z}e^{-\pi x^2/r^2}}{\sum\limits_{x\in \Z}e^{-\pi x^2/r^2}}\ = \frac{\sum\limits_{x\in \Z}e^{-\pi (q/r)^2x^2}}{r\sum\limits_{x\in \Z}e^{-\pi r^2x^2}}$$
    where the last equality holds by applying the Poisson summation formula on the denominator. When $1<r<q$, both summations in the numerator and denominator can be bounded by constants. Thus $\Psi_{\Z_q,r}(0) = \Theta(1/r)$.
    \subsection{Estimating $\E[\chi(\Psi_{\Z_q,r})]$}\label{app:rho for DG}
    For $e\leftarrow\Psi_{\Z_q,r}$, we then have $\rho = \E[\chi(e)]= \E[\omega_q(e)]$. This means that $r\in[0,\infty]$ is the (unique) solution to the equation
\begin{align}
    \rho &= \E[\eta_i] = \frac{\sum_{j\in\Z_q}\omega_q^j\sum_{x\in q\Z+j}{e^{-\pi x^2/r^2}}}{\sum_{x\in\Z}{e^{-\pi x^2/r^2}}} \\
    &= \frac{\sum_{x\in\Z}{\omega_q^xe^{-\pi x^2/r^2}}}{\sum_{x\in\Z}{e^{-\pi x^2/r^2}}}  = \frac{\sum_{x\in\Z}{\cos\left(\frac{2\pi x}{q}\right)e^{-\pi x^2/r^2}}}{\sum_{x\in\Z}{e^{-\pi x^2/r^2}}}\label{rho equation for Gaussian distribution}
\end{align}

    Let us obtain a lower bound for \eqref{rho equation for Gaussian distribution}. 
    For $g(x) = \cos\left(\frac{2\pi x}{q}\right)e^{-\pi x^2/r^2}$,\\ $\hat{g}(y) = \frac{r}{2} \left( e^{-\pi r^2 \left(y - \frac{1}{q}\right)^2} + e^{-\pi r^2 \left(y + \frac{1}{q}\right)^2} \right)$ and for $h(x) = e^{-\pi x^2/r^2}$, $\hat{h}(y) = re^{-\pi r^2y^2}$. Using these in \eqref{rho equation for Gaussian distribution}, we obtain
    \begin{align*}
        \rho = \frac{\sum_{y\in\mathbb{Z}} \hat{g}(y)}{\sum_{y\in\mathbb{Z}} \hat{h}(y)} = \frac{e^{-\pi r^2/q^2}\sum_{y\in\mathbb{Z}}\cosh{\left(\frac{2\pi r^2 y}{q}\right)}e^{-\pi r^2y^2}}{\sum_{y\in\mathbb{Z}} e^{-\pi r^2 y^2}}
        \geq e^{-\pi r^2/q^2}     
    \end{align*}

In our results, we see that the sample complexity required for both distinguishing methods using Kikuchi graphs depends inverse polynomially on $\rho$. Since the width $r$ for the LWE distribution is chosen to be less than $q$, this dependency will only contribute a constant term.

\section{A quantum speedup}\label{app:quantum speedup}

\textbf{Some Notation.} $\bm{x}^\dagger$ denotes the conjugate transpose of a vector $\bm{x}$. $\ket{\bm{x}}$ denotes the unit vector in the direction of $\bm{x}$ and $\bra{\bm{x}} := \ket{\bm{x}}^\dagger$. For a matrix $A$, $\norm{A}_{\max} := \max_{i,j}\abs{A_{ij}}$.\\

In essence, to solve decisional $\LIN{k}{R}$, the Kikuchi method performs estimation of the spectral norm of a sparse Hermitian matrix in a space of dimension $N = \binom{n}{l}(\abs{R}-1)^l$. For any known classical method, such as the power iteration, eigenvalue estimation requires matrix multiplications, so it is unlikely for the time complexity to go below $\widetilde{\Omega}(N)$. However, with quantum algorithmic techniques, we can do better. 

A Hamiltonian $H$ on $t$ qubits is a Hermitian matrix of dimension $2^t$. $H$ is called \emph{$s$-sparse} if each row and each column of $H$ has at most $s$ non zero entries. Throughout the discussion, we shall assume the normalization constraint that $\|H\|_{\max} \leq 1$. Approximating the spectral norm of a sparse Hamiltonian is well studied in quantum computing. Indeed, for a local Hamiltonian (which is a special case of sparse hamiltonian), this problem is QMA complete. However, if we are provided with a \emph{guiding} state $\ket{\gamma}$, i.e. a state that has good overlap with a unit eigenvector $\ket{v}$ corresponding to the maximum absolute eigenvalue, i.e. $(\braket{\gamma |v})^2$ is large, the problem can be easier. If we are provided with a guiding state having $1/\poly(t)$ overlap with the leading eigenvector, then Gharibian and Le Gall \cite{guided_sparse_hamiltonian_BQP_complete} show that the problem is BQP-complete. 

Schmidhuber et al. \cite{quartic_speedup} provide a quantum algorithm to solve the Guided Sparse Hamiltonian problem by combining standard quantum algorithmic techniques of Hamiltonian simulation, phase estimation and amplitude estimation. We define the problem and state their result below. For details, the reader may refer to their paper.
\subsection{Guided Sparse Hamiltonian Problem}

The input Hamiltonian $H$ is assumed to be given in the form of quantum oracle access by providing oracles $O_H$ and $O_F$. $O_H$ is called the \emph{adjacency matrix} oracle, which on input $(i,j)$  (for $i,j\in[2^t]$) gives the entry $H_{ij}$ of $H$ up to rounding. $O_F$ is the \emph{adjacency list} oracle, which on input $(i,r)$ (for $i\in[2^t], r\in[s]$) either returns the index of the $r$\textsuperscript{th} non-zero entry in the $i$\textsuperscript{th} row, or indicates that there are less than $r$ non-zero entries in $i$\textsuperscript{th} row. 

The advantage of having oracle access to $H$ in this manner is that we do not have to store the Hamiltonian $H$ that may possibly be of large dimension; rather we can access it ``on the fly'' when required.

\begin{definition}\label{defn GSHS} (Definition 62 in \cite{quartic_speedup})
    In the GSHS problem, we are given the following as input : 
    \begin{enumerate}
        \item An $s$-sparse hamiltonian $H$ on $t$ qubits with $\|H\|_{\max} \leq 1$ as specified with the oracles $O_H$ and $O_F$.
        \item A quantum circuit using $G$ gates which performs the mapping : $\ket{0^t}\ket{0^A} \mapsto \ket{\Psi}\ket{0^A}$.
        \item Parameters $\lambda \in \mathbb{R}, \alpha \in (0,1), \gamma \in (0,1]$. 
    \end{enumerate}
    To solve it, one must output : 
    \begin{itemize}
        \item $\YES$, if $\|\Pi_{\geq \lambda} (H)\ket{\Psi}\| \geq \gamma$ (which also implies $\norm{H}\geq\lambda$)
        \item $\NO$, if $\|H\| \leq (1 - \alpha)\lambda$
    \end{itemize}
\end{definition}

\begin{theorem}\label{solving GSHS} (Theorem 63 in \cite{quartic_speedup})
The Guided Sparse Hamiltonian problem can be solved with high probability by a quantum algorithm that uses: 
\begin{itemize}
    \item \((Q=\widetilde{O}\bigl(s/\gamma\alpha\lambda)\) queries to the oracles
    \item \(\widetilde{O}\bigl(G/\gamma+\operatorname{polylog}(Q)/\gamma+Qt\bigr)\) gates
    \item \((A+O(t)+\widetilde{O}(\log Q)\) qubits
\end{itemize}
\end{theorem}

For a $\LIN{k}{R}$ instance $\I$ of dimension $n$ and $m$ samples, the associated Hermitian matrix $\K_l(\I)$ has the following parameters associated with it:
\begin{enumerate}
    \item $2^t = N = \binom{n}{l}(q^l$, so $t = O(l\log n)$
    \item Its sparsity $s=\Delta = \frac{m\theta^{(n,l,k)}}{\binom{n}{l}(q^l} = O(\log N) = O(l\log n)$ for the separation in spectral norm for random vs planted case (from Propositions \ref{delta' lower bound random case}, \ref{delta' lower bound planted case})
    \item $\lambda > (1+O(1))\rho\Delta$, say $\lambda = 3/2\rho\Delta$ (lower bound of spectral norm for planted case)
    \item $(1-\alpha)\rho\Delta = (1-O(1))\rho\Delta$, say $\alpha = 1/2$ (upper bound of spectral norm for random case)
\end{enumerate}
Additionally, the oracles $O_H$ and $O_F$ can be efficiently computed using Construction \ref{kLIN to 2LIN}. 

From point 3. of Definition \ref{defn GSHS}, $\gamma^2$ is the parameter which corresponds to the overlap of the guiding state $\ket{\Psi}$. If $\gamma$ is less than inverse polynomial, then from points 1 and 2 of Theorem \ref{solving GSHS}, $1/\gamma$ is the dominating factor for the time (i.e. gate) complexity. 

Note that if $\ket{\Psi}$ is a random unit vector, then its overlap (with high probability) to the cutoff eigenspace (corr. to $\Pi_{\geq\lambda}$ that in the worst case is of small dimension) will be inversely proportional to the dimension with high probability, i.e. 
\[
\gamma^2 = \Omega\left( \frac{1}{\binom{n}{l}(\abs{R}-1)^l} \right)
\]
in our setting. This already yields a quantum algorithm with complexity $\frac{\text{poly}(n)}{\gamma} = \text{poly}(n)\cdotp\sqrt{\binom{n}{l}(\abs{R}-1)^l}$, a (nearly) quadratic speedup over the classical algorithm. 

The other key observation of \cite{quartic_speedup} is the ability to prepare a better guiding state, one which has overlap of $\Omega\left( \frac{1}{\sqrt{\binom{n}{l}(\abs{R}-1)^l}}\right)$ with the cutoff-eigenspace.

\subsection{Intuition for Guiding State} 
Again, we give a brief intuition behind the better guiding state; the details of the proofs will be analogous to Appendix C of \cite{quartic_speedup}. 

Suppose for an instance $\I = \{(\bm{a}_i,b_i)\}_{i\in[m]}$ we construct $l/k$ copies of the following state: 

$$ \ket{\gamma(\I)} := \frac{1}{\sqrt{m}} \sum_i \chi(b_i) \ket{\bm{a}_i} \in \mathbb{C}^{R^n_k} $$

This state correlates well (better than random) with the planted assignment $\ket{\bm{s}} = \frac{1}{\sqrt{n}}\sum_i \chi(s_i)\ket{i}$ : 

$$|\braket{\bm{s}^{\otimes \ell}| \gamma(\I)^{\otimes \ell/k}}|^2 = \left( \braket{\bm{s}^{\otimes k}| \gamma(\I)} \right)^{2 \ell / k} \approx \left(\frac{\rho\cdotp m}{n^kq^k}\right)^{\ell / k}$$

If $m \approx \Theta(((\abs{R}-1)n)^{k/2})$ then the overlap $\sim ((\abs{R}-1)n)^{-l/2}$, which is a quadratic improvement from that of a random guiding state. The actual guiding state is approximately a symmetrized and a normalized version of $\ket{\gamma(\I)}$. Also, the other technical challenge is to show a high overlap with the cutoff eigenspace of $\K_l(\I)$ rather than with $\bm{s}$. These issues are handled in detail in Appendix C of \cite{quartic_speedup}.

\end{document}